\documentclass[11pt,a4paper]{article}
\usepackage[utf8]{inputenc}
\usepackage{amsmath}
\usepackage{amsthm}
\usepackage{amssymb}
\usepackage{graphicx}
\usepackage{color}
\usepackage{hyperref}

\usepackage{csquotes}
\usepackage[T1]{fontenc}
\usepackage{amsfonts}
\usepackage{mathrsfs}
\usepackage{mathtools}

\usepackage{physics}
\usepackage[tmargin=3cm,lmargin=2.5cm, rmargin=2.5cm]{geometry}
\usepackage{tikz}
\usepackage[all]{xy}
\usepackage{verbatim}

\input xy
\xyoption{all} \tolerance=500

\newtheoremstyle{base}
    {}
    {}
    {}
    {}
    {}
    {}
    { }
    {\textbf{\thmname{#1}\thmnumber{ #2}.\thmnote{ (#3)}}}

\theoremstyle{base}
\newtheorem{theorem}{Theorem}[section]

\newtheorem{proposition}[theorem]{Proposition}

\theoremstyle{remark}
\newtheorem{remark}[theorem]{Remark}
 
\newdir{|>}{!/4.5pt/\dir{|}
    *:(1,-.2)\dir^{>}
    *:(1,+.2)\dir_{>}}
\newdir{ (}{{}*!/-5pt/@^{(}}

\makeatletter
\newsavebox{\@brx}
\newcommand{\llangle}[1][]{\savebox{\@brx}{\(\m@th{#1\langle}\)}%
    \mathopen{\copy\@brx\kern-0.5\wd\@brx\usebox{\@brx}}}
\newcommand{\rrangle}[1][]{\savebox{\@brx}{\(\m@th{#1\rangle}\)}%
    \mathclose{\copy\@brx\kern-0.5\wd\@brx\usebox{\@brx}}}
\makeatother

\newcommand{\R}{\mathbb{R}}

\newcommand{\alg}{\mathfrak{g}}

\newcommand{\T}{\mathsf{T}}
\newcommand{\Ver}{\mathsf{V}}
\newcommand{\Hor}{\mathsf{H}}
\newcommand{\X}{\mathfrak{X}}
\newcommand{\V}{\mathcal{V}}
\newcommand{\dif}{\mathrm{d}}
\newcommand{\der}{\frac{\dif}{\dif t}\at[\Big]{t=0}}

\newcommand{\id}{\mathrm{id}}

\newcommand{\pr}{\mathrm{pr}}
\newcommand{\adP}{{\mathrm{ad}P}}
\newcommand{\0}{\textbf{0}}
\newcommand{\AdG}{\mathrm{Ad}}
\newcommand{\ad}{\mathrm{ad}}
\newcommand{\Ker}{\mathrm{Ker}}

\newcommand{\fr}{{P_{S^2}}}

\newcommand{\quo}{{\tilde{\delta}_{A(P)}}}
\newcommand{\quoD}{{\tilde{\delta}_{A^*(P)}}}

\newcommand{\spacesmall}{2mm}
\newcommand{\spacemid}{5mm}

\def\sT{\mathsf{\T}}

\newcommand{\at}[2][]{#1|_{#2}}

\title{Tulczyjew triple on the Atiyah algebroid with connection}
\author{Katarzyna Grabowska, Paweł Korzeb\footnote{Currently a student at the Ludwig Maximilian University of Munich.}, Kuba Krawczyk\footnote{Currently a student at the University of Amsterdam, Institute of Physics.} \\ \\
{} {\it Faculty of Physics}\\
      {\it University of Warsaw}}

\date{\today}

\begin{document}

\maketitle

\begin{abstract}
    \noindent
    The Tulczyjew triple on a principal bundle with connection is constructed in a convenient trivialisation. A reduction by the structure group is performed leading to the triple on the trivialised Atiyah algebroid and a presentation of this algebroid via a double vector bundle morphism. The dynamics of physical systems with configuration manifolds having the structure of a principal bundle with connection or the Atiyah algebroid is discussed and applied to the example of an axially symmetric body confined to a sphere.
    
    
\end{abstract}

\section{Introduction}

Lie algebroid is a geometric object very extensively described in the literature. First, it appeared in the 1967 works of Pradines \cite{Pradines67} as an infinitesimal object for a Lie groupoid. In 1996, Weinstein \cite{Weinstein96} and Liebermann \cite{Liebermann96} used both structures -- Lie groupoids and Lie algebroids -- in analytical mechanics. The structure of Lie algebroid in Lagrangian mechanics plays a similar role as symplectic or Poisson structure in Hamiltonian mechanics, i.e. it provides tools for generating relevant equations from a Lagrangian or, more generally, from a Lagrangian generating object.

The Tulczyjew triple is a structure used in geometric analytical mechanics introduced by Tulczyjew in his various works (ex. \cite{WMT76a, WMT76b, Tulczyjew}). It is a commutative diagram consisting of maps essential in both Lagrangian and Hamiltonian descriptions of physical systems.

In \cite{Grabowski99}, Grabowski and Urbański showed that a standard Lie algebroid structure, or its generalised version ({\sl general algebroid}), is equivalently encoded in a certain double vector bundle morphism. Together with Grabowska, they showed \cite{Grabowska06} that mechanics on Lie algebroids may be described by the Tulczyjew triple, and in that case its defining morphisms encode the base algebroid.

In this paper, we closely examine the case where the base configuration manifold is the well-known {\sl Atiyah Lie algebroid}, which additionally enjoys the presence of a connection on its principal bundle.

The paper is organised as follows. In Sections 2.1 and 2.2, we recall some necessary elements of the general theory of the Tulczyjew triple and Lie algebroids. In Section 2.3, we briefly present the construction and reduction of the Tulczyjew triple on a Lie group (see \cite{triple_lie_group}), which is instructive for the rest of this work. In Section 3, we introduce some basic notions related to principal bundles and connection theory. Then, Section 4 contains the first part of the results of this paper. There, we construct the Tulczyjew triple on a principal bundle and offer its convenient trivialisation when a connection is present. We end the section by calculating the trivialised dynamics for a configuration manifold being a principal bundle with connection, and we apply it to a physical example -- an axially symmetric body confined to a sphere. In Section 5, we reduce all the previous constructions by the action of the tangent group to obtain the, conveniently trivialised, Tulczyjew triple on the Atiyah algebroid, including the morphism encoding its Lie structure and the dynamics of systems modelled by it.

Part of the results presented in Section 4 have already been shown by Esen, Kudeyt and Sütlü \cite{Esen}, but they used different methods and then discussed a completely separate reduction by the action of the structure group (not the tangent of the group) and hence made no connection to the Atiyah algebroid.

\section{Tulczyjew triple and algebroids}

The original Tulczyjew triple is a diagram containing all the information necessary for deriving dynamics both from Lagrangian or Hamiltonian generating objects and understanding the nature of the Legendre transformation. For $M$ being a configuration manifold, the classical form of the Tulczyjew triple is as follows:
\begin{equation}\label{e:1}
\xymatrix@C-15pt{
\sT^\ast\sT^\ast M\ar[dr]_{\pi_{\sT^\ast M}} & & \sT\sT^\ast M\ar[rr]^{\alpha_M}\ar[ll]_{\beta_M}\ar[dr]^{\sT\pi_M}\ar[dl]_{\tau_{\sT^\ast M}} & & \sT^\ast\sT M \ar[dl]^{\pi_{\sT M}} \\
& \sT^\ast M\ar[dr]^{\pi_M} & & \sT M\ar[dl]_{\tau_M} &  \\
& & M & & .
}
\end{equation}
The map $\beta_M$ is associated with the canonical symplectic form $\omega_{M}$ on $\sT^\ast M$ -- it is just a contraction of a vector from $\sT\sT^\ast M$ with this form. The map $\alpha_M$ is the Tulczyjew isomorphism introduced in \cite{WMT76b}. This simple diagram works for autonomous mechanics, i.e. when the Lagrangian or Hamiltonian does not depend explicitly on time. It can be used also in relativistic mechanics where $M$ is the Minkowski space-time, and there is no distinguished time.

The three top bundles in the Tulczyjew triple have rich geometric structures. Each
of them is a double vector bundle, i.e. it has two compatible vector bundle structures (see
\cite{higher_vector_bundles, double_bundles_duality}). We also note that having one of the maps, we may canonically obtain the other one via Dufour's canonical isomorphism \cite{canonical_gamma} as $\beta_M = \gamma_{\T M}^{-1} \circ \alpha_M$.

In the following two sections, we shall present shortly the mechanical aspects of the triple and its mathematical significance in the context of algebroids. We shall finish this section with an instructive example -- the Tulczyjew triple in the case where the base manifold is a Lie group, and we use the left translation to trivialise tangent and cotangent iterated bundles.

\subsection{Tulczyjew triple in mechanics}

Tulczyjew's approach to mechanics, based on a profound analysis of the variational approach to physical theories \cite{Tulczyjew}, has proven very fruitful. The Tulczyjew triple is only one of the results of this analysis. Originally it was designed to make sense of the Legendre transformation for Lagrangian systems with singular Lagrangians. However, it turned out that the triple can be easily modified to different physical situations. There exists the Tulczyjew triple for time-dependent mechanics in non-relativistic cases, i.e. when Lagrangian or Hamiltonian depend explicitly on time. There exists also the Tulczyjew triple formulated in the language of the geometry of affine values, which can be used, for example, to describe mechanical systems in Newtonian space-time in a way which is independent of the choice of a particular observer. Tulczyjew's approach serves well also in the case of constrained systems and systems with symmetry.

In its infinitesimal version, a mechanical system in Tulczyjew's vision is described by phase equations called {\sl the dynamics}. The dynamics is a subset of vectors tangent to the phase space. In many cases, this subset is an image of a vector field, but there are other situations, e.g. free relativistic particle where implicit equations are needed (other examples can be found in \cite{slowandcareful}).

To illustrate the use of the Tulczyjew triple, we shall use the classical diagram (\ref{e:1}). Let $M$ be the manifold of positions of the mechanical system we describe. In autonomous mechanics, a Lagrangian is a function depending on positions and velocities, i.e. it is a function on the tangent bundle $L:\sT M\rightarrow\R$. The image of the differential $\dif L(\sT M)$ is a Lagrangian submanifold in $\sT^\ast\sT M$, and the dynamics $\mathcal{D}$ is given by the formula
$$\mathcal{D}=\alpha_M^{-1}(\dif L(\sT M)).$$
In local coordinates $(q^i)$ on $M$, $(q^i, p_j)$ on the phase space $\sT^\ast M$ and $(q^i, p_j, \dot q^k, \dot p_l)$ on $\sT\sT^\ast M$, the equations for $\mathcal{D}$ read
$$p_i=\frac{\partial L}{\partial \dot q^i}, \qquad \dot p_j=\frac{\partial L}{\partial q_j}.$$
It is easy to see that the Euler-Lagrange equation with external forces equal to zero is a differential consequence of the dynamics. The dynamics contains then the definition of momenta, as well as, the equations of motion. In more complicated physical situations, the Lagrangian submanifold  $\dif L(\sT M)$ can be replaced by a Lagrangian submanifold generated by a more complicated generating object, for example, a function on a submanifold (e.g. for vaconomic mechanics) or a generating  family (e.g. for massless particles in relativistic setting). The full diagram of the Lagrangian side of the Tulczyjew triple, that includes the double vector bundle structure of iterated tangent and cotangent bundles as well as the Legendre map $\lambda_L$, reads
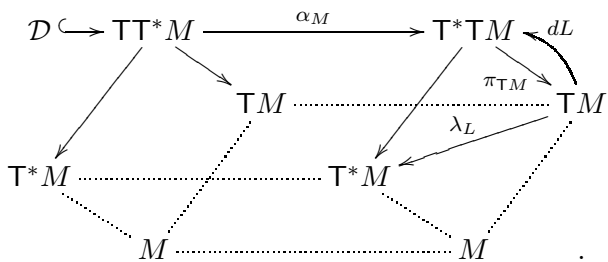
\begin{equation*}\label{e:2}
\xymatrix@C-15pt@R-10pt{
{\mathcal{D}}\ar@{ (->}[r]& \sT\sT^\ast M \ar[rrr]^{\alpha_M} \ar[dr]\ar[ddl]
 & & & \sT^\ast\sT M\ar[dr]_{\pi_{\sT M}}\ar[ddl] & \\
 & & \sT M\ar@{.}[rrr]\ar@{.}[ddl]
 & & & \sT M \ar@{.}[ddl]\ar@/_1pc/[ul]_{d L}\ar[dll]_{\lambda_L}\\
 \sT^\ast M\ar@{.}[rrr]\ar@{.}[dr]
 & & & \sT^\ast M\ar@{.}[dr] & &  \\
 & M\ar@{.}[rrr]& & & M & .
}
\end{equation*}

The Hamiltonian side of the Tulczyjew triple is associated with the symplectic form $\omega_M$ and the map $\beta_M\sT\sT^\ast M\rightarrow \sT^\ast\sT^\ast M$, which is defined as  $\beta_M(v)=\omega_M(v,\cdot)$. The formula for generating the dynamics from a Hamiltonian reads
$$\mathcal{D}=\beta_M^{-1}(\dif H(\sT^\ast M)).$$
This formula is just another way of writing the definition of the Hamiltonian vector field for a function $H:\sT^\ast M\rightarrow \R$. This is evident when we write the equations defining $\mathcal D$ in coordinates:
$$\dot q^i=\frac{\partial H}{\partial p_i}, \qquad \dot p_j=-\frac{\partial H}{\partial q^j}.$$
The Lagrangian submanifold $\dif H(\sT^\ast M)$ can, however, be replaced by a Lagrangian submanifold generated by another Hamiltonian generating object, as it is, for example, in the case of a free relativistic particle. Using Lagrangian submanifolds instead of vector fields and generating families instead of just functions allows for both Lagrangian and Hamiltonian description of most mechanical systems. It works even for those with singular Lagrangians or no Lagrangians at all, as a massless particles. The full diagram of the Hamiltonian side of the Tulczyjew triple, that includes the double vector bundle structure of iterated tangent and cotangent bundles, reads
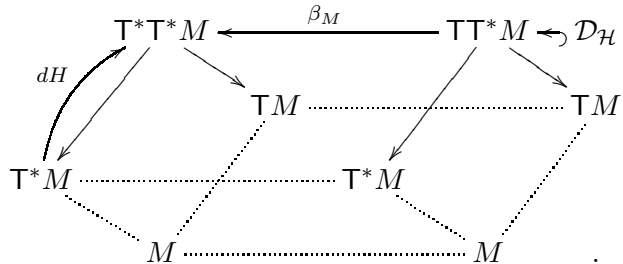
\begin{equation*}\label{e:3}
\xymatrix@C-15pt@R-10pt{
 & \sT^\ast\sT^\ast M  \ar[dr] \ar[ddl]
 & & & \sT\sT^\ast M\ar[dr]\ar[ddl] \ar[lll]_{\beta_M}&
 { \mathcal{D_H}}\ar@{ (->}[l] \\
 & & \sT M\ar@{.}[rrr]\ar@{.}[ddl]
 & & & \sT M \ar@{.}[ddl]\\
 \sT^\ast M\ar@{.}[rrr]\ar@{.}[dr] \ar@/^1pc/[uur]^{d H}
 & & & \sT^\ast M\ar@{.}[dr] & &  \\
 & M\ar@{.}[rrr]& & & M & .
}
\end{equation*}

\subsection{Algebroids via Tulczyjew triple morphisms}
\label{sec:algebroids_and_triples}

According to the most widely-used definition, a {\it Lie algebroid} is a vector bundle $\tau:E\rightarrow M$ together with a Lie bracket on sections of the bundle and a vector bundle morphism $\rho: E\rightarrow\sT M$, called {\sl the anchor}, satisfying the condition
$$[X,fY]=f[X,Y]+\rho(X)(f)Y$$
for all sections $X,Y$ of the bundle $\tau$ and all smooth functions $f$ on $M$. The first canonical example of a Lie algebroid is the tangent bundle $\sT M$ with the Lie bracket of vector fields and identity as the anchor. Another example is a Lie algebra $\alg$ over a one-point manifold with the Lie algebra bracket and the anchor that maps every element of $\alg$ to the zero vector. The Atiyah algebroid of a principal bundle that will be discussed in detail later is another canonical example of the structure.

It is well known that the presence of a Lie algebroid structure on $\tau: E\rightarrow M$ is equivalent to the existence of a linear Poisson structure on the dual bundle $\pi: E\rightarrow M$. Then, it turns out that every linear Poisson structure $\Lambda$ on $E$ may be written as a unique double vector bundle morphism $\varepsilon : \T^* E \rightarrow \T E^*$ over the identity on $E^\ast$ with the following diagram
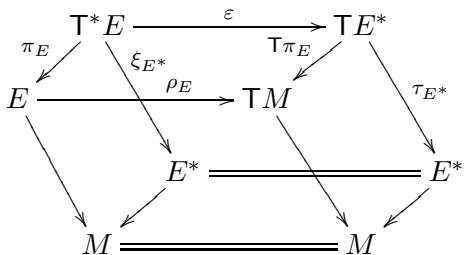
\begin{equation*}\begin{gathered}
    \label{xy:lie_algebroid}
    \xymatrix@C-15pt@R-10pt{
         & \T^*E \ar[rrr]^{\varepsilon} \ar[dl]_{\pi_{E}} \ar[ddr]^(0.3){\xi_{E^*}} & & & \T E^*  \ar[dl]_{\T\pi_E} \ar[ddr]^{\tau_{E^*}} & \\
        E \ar[rrr]^(0.65){\rho_E} \ar[ddr] & & & \T M \ar[ddr] & & \\
         & & E^* \ar@{=}[rrr] \ar[dl] & & & E^* \ar[dl] \\
         & M \ar@{=}[rrr] & & & M &
    }
\end{gathered}\end{equation*}
via the canonical Dufour's isomorphism as $\tilde{\Lambda} = \varepsilon \circ \gamma_E$. In the case of the canonical algebroid $E=\sT M$, this Poisson structure is just given by the canonical symplectic form. The corresponding double vector bundle morphism is the inverse of the Tulczyjew isomorphism $\varepsilon = \alpha_M^{-1}$. Further study of systems with symmetry \cite{Grabowska06} led to the conclusion that the Tulczyjew triple is not only a useful tool in classical mechanics but carries a lot of information on the structure of the three bundles that compose it. In particular, reductions by the symmetric degrees of freedom allow us to construct the double vector bundle morphism descriptions of other algebroids.

\subsection{Tulczyjew triple on a Lie group}

In this section, we briefly present the Tulczyjew triple on a Lie group and its reduction following Zając and Grabowska \cite{triple_lie_group}. It serves a dual purpose. First, it is an example of a reduction mentioned above. Second, the construction of the triple on a principal bundle carried out in Section 4 may be thought of as a generalisation of the Lie group case in an obvious way.

The construction is based on the standard Lie group isomorphism $\T G \simeq G \times \alg : (g, v_g) \mapsto (g, \T l_{g^{-1}}(v_g))$, where $\alg$ is the Lie algebra of $G$ and $l_g(x) = gx$. It allows us to trivialise all iterated tangent and cotangent bundles as further products of $\alg$ and its dual $\alg^*$. The Tulczyjew triple takes then the trivialised form
\begin{equation*}\label{eq:triple_lie_group}
\xymatrix@C-15pt{
G \times \alg^* \times \alg^* \times \alg \ar[dr]_{(\pr_1, \pr_2)} & & G \times \alg^* \times \alg \times \alg^* \ar[rr]^{\tilde{\alpha}_G}\ar[ll]_{\tilde{\beta}_G}\ar[dr]^{(\pr_1, \pr_3)}\ar[dl]_{(\pr_1, \pr_2)} & & G \times \alg \times \alg^* \times \alg^* \ar[dl]^{(\pr_1, \pr_2)} \\
& G \times \alg^* \ar[dr]^{\pr_1} & & G \times \alg\ar[dl]_{\pr_1} &  \\
& & G & & ,
}
\end{equation*}
where $\tilde{\alpha}_G(g, A, X, B) = (g, X, B - \ad_X^*(A), A)$ and $\tilde{\beta}_G(g, A, X, B) = (g, A, -B + \ad_X^*(A), X)$.

The left group action on itself lifts naturally to the tangent bundle $\T G \simeq G \times \alg$, and one can now consider the quotient space $\T G /_G \simeq \alg$ as the base of the triple. Using the induced actions on iterated tangent and cotangent bundles, we get the reduced Tulczyjew triple on $G$:
\begin{equation*}\label{eq:triple_lie_group_reduced}
\xymatrix@C-10pt{
\alg^* \times \alg \ar[dr]_{\pr_1} & & \alg^* \times \alg^* \ar[rr]^{\tilde{\alpha}_\alg}\ar[ll]_{\tilde{\beta}_\alg}\ar[dl]^{\pr_1} & & \alg \times \alg^* \ar[dl]^{\pr_1} \\
& \alg^* & & \alg & .
}
\end{equation*}
$\tilde{\alpha}_\alg$, recovering the full double vector bundle structure, describes the Lie algebroid structure of the Lie algebra $\alg$:
\begin{equation*}\begin{gathered}
    \xymatrix@C-15pt@R-10pt{
         & \alg \times \alg^* \ar[rrr]^{\varepsilon = \tilde{\alpha}_\alg^{-1}} \ar[dl]_{\pr_1} \ar[ddr]^(0.3){\pr_2} & & & \alg^* \times \alg^*  \ar[dl] \ar[ddr]^{\pr_1} & \\
        \alg \ar[rrr]_(0.35){\rho_\alg = 0} \ar[ddr] & & & \0 \ar[ddr] & & \\
         & & \alg^* \ar@{=}[rrr] \ar[dl] & & & \alg^* \ar[dl] \\
         & \{\bullet\} \ar@{=}[rrr] & & & \{\bullet\} & .
    }
\end{gathered}\end{equation*}

This paper aims to perform a similar construction and reduction on the slightly richer structure of a principal bundle with connection and, as a result, obtain the explicit form of the double vector bundle morphism encoding the trivialised Atiyah algebroid.

\section{Principal bundles and connections}

In this section, we briefly recall the geometric notions of a principal bundle and a connection, which will serve us as the base manifold of the Tulczyjew triple.

A \textit{principal bundle} is a differential manifold $P$ with a Lie group $G$, $n := \dim G$, acting on it freely and properly from the right $P\times G\ni (p,g)\mapsto \wp_g(p) \equiv pg\in P$. Properties of the action ensure that the space of its orbits is also a differential manifold, which we often denote with $M$, $m :=\mathrm{dim} M$. Therefore, the principal bundle is a fibre bundle $(P, M, G, \pi)$.

The \textit{vertical bundle} $(\Ver P, P, \R^n, \pi_{\Ver P})$ of $P$ is a bundle, in which the fibre over any point $p \in P$ consists of the \textit{vertical vectors}, i.e. elements of $\T_p P$ belonging to the kernel of $\T \pi$.

We will also use the \textit{adjoint bundle} $(\adP, M, \alg, \pi_\adP)$ of a principal bundle $P$, i.e. the bundle associated with $P$ via the adjoint representation $\AdG$ of $G$ on its Lie algebra $\alg$. One may obtain the adjoint bundle by dividing the vertical bundle $\Ver P$ by $G$. Additionally, we introduce a bracket on $\adP$ inherited from the Lie bracket of $\alg$: $\big[[(p, X)], [(p, Y)]\big]_\adP := [(p, [X, Y]_\alg)]$.

An \textit{Ehresmann connection} $(\Hor P, P, \R^m, \pi_{\Hor P})$ is a vector subbundle of the tangent bundle of $P$ complementary to the vertical subbundle $\Ver P$ over every point of the base, i.e. $\T P \simeq \Hor P \oplus_P \Ver P$. Additionally, the connection on a principal bundle is \textit{principal} if the following condition is satisfied:
\begin{equation*}
    \label{eq:ehr_pr_con}
    \forall \, p \in P, \, g \in G \qquad \T_p\wp_g(\Hor_p P) = \Hor_{pg}P.
\end{equation*}
In the next sections, we will assume all connections are principal.

Any connection gives rise to a projection onto $\Hor P$ along $\Ver P$, which we denote with $\pr_\Hor : \T P \longrightarrow \Hor P$, and a collection of maps $\T_x M\rightarrow \sT_p P$ for all pairs $(x,p)$ such that $x=\pi(p)$, called the \textit{horizontal lift} $\T_{\pi(p)} M \ni v \mapsto v_p^H \in \Hor_p P \subset \T_p P$ providing an isomorphism $\Hor_pP \simeq \T_{\pi(p)} M$. Therefore, there is an identification
\begin{equation}
    \label{eq:iso_horizontal}
    \Hor P \simeq P \times_M \T M.
\end{equation}

If the bundle is principal, the fibres are isomorphic to the structure Lie group. The $G$-equivariant \textit{fundamental vector fields} $\sigma(X)_{p}:=\der \, (pe^{tX})$ induce, via the mapping $P \times \alg \ni (p, X) \mapsto \sigma(X)_p \in \Ver P$, a trivialisation
\begin{equation}
    \label{eq:iso_vertical}
    \Ver P \simeq P \times \mathfrak{g}.
\end{equation}
This observation inspires an equivalent way of defining a connection for principal bundles. A \textit{connection one-form} on a principal bundle $P$ is a $\mathfrak{g}$-valued one-form $\omega\in\Omega(P,\mathfrak{g})$, which satisfies the condition
\begin{equation*}
    \label{eq:pr_con_form_1}
    \omega(\sigma(X)) = X,
\end{equation*}
and is $G$-equivariant:
\begin{equation*}
    \label{eq:pr_con_1form}
    {\wp_g}^*\omega = \AdG_{g^{-1}} \circ \omega.
\end{equation*}
A principal Ehresmann connection defines the connection one-form $\omega := (\sigma_p)^{-1} \circ (\id_{\T P} - \pr_\Hor)$, and a connection one-form defines the Ehresmann connection $\Hor P := \ker \omega$.

Any principal connection has the associated \textit{curvature two-form} $\Omega$ given by
\begin{equation*}
    \Omega :=\pr_\Hor^*\dif\omega \equiv \dif\omega \circ (\pr_\Hor \times \pr_\Hor) \in \Omega^{2}(P,\mathfrak{g}).
\end{equation*}
This two-form may be also expressed via:
\begin{itemize}
    \item when $V, W \in \X(P)$
        \begin{equation} \label{alternative_omega}
            \Omega(V, W) = -\omega([\pr_\Hor V, \pr_\Hor W]),
        \end{equation}
    \item ({\sl the structure equation})
        \begin{equation*} \label{structure_equation}
            \Omega(V,W)=\dif\omega(V,W)+[\omega(V),\omega(W)].
        \end{equation*}
\end{itemize}

\begin{remark}
    Because the value of $\Omega(V, W)$ depends only on the horizontal parts of $V, W \in \X(P)$, we often abuse the notation by writing $\Omega(v, w)$ for $v, w \in \X(M)$ when it does not lead to confusion. Then,
    \begin{equation}
        \label{eq:curv_abuse}
        \Omega(v, w) \equiv \Omega(v^\Hor, w^\Hor) = [v^\Hor, w^\Hor] - [v, w]^\Hor.
    \end{equation}
\end{remark} 

\section{Tulczyjew triple on a principal bundle with connection}
    Hereafter, we consider a principal bundle $P$ with a principal connection $\Hor P$. In this section, we start by recalling how the presence of the connection allows to trivialise $\T P$. Then, we will use it to further trivialise iterated tangent and cotangent bundles and the canonical double vector bundle morphisms between them.

\subsection{Trivialisation of the tangent and cotangent of a principal bundle}
    The connection on the bundle provides us with a decomposition $\T P \simeq \Hor P \oplus_P \Ver P$. Using the isomorphisms \eqref{eq:iso_vertical} and \eqref{eq:iso_horizontal}, we further rewrite it to $\T P \simeq (P \times_M \T M) \oplus_P (P \times \alg)$. Omitting the redundant information, we obtain (slightly abusing the notation)
    \begin{equation}
        \label{eq:triv_TP}
        \imath_{\T P} : \T P\simeq P\times_{M} \T M\times\mathfrak{g}.
    \end{equation}

    \begin{remark}
        The formula above is explicitly given by
        \begin{equation*}
            \imath_{\T P}(V) = \left(\tau_P(V), \T\pi(V), \omega(V)\right) \in P \times_{M} \T M \times \mathfrak{g},
        \end{equation*}
        where $\tau_P : \T P \rightarrow P$, $\pi : P \rightarrow M$ and $\omega$ is the connection one-form of $\Hor P$. We also note that the inverse mapping is given by
        \begin{equation*}
            \label{eq:triv_TP_inv}
            \imath_{\T P}^{-1}(p, \V, X) = \V^H_p + \sigma(X)_p \in \T P.
        \end{equation*}
    \end{remark}

    \vspace{\spacesmall}

    Making use of \eqref{eq:triv_TP}, as well as the isomorphism $(\Ver_p P\times \Hor_p P)^{*}\simeq (\Ver_p P)^{*}\times (\Hor_p P)^{*}$ over every $p \in P$, we have
    \begin{equation*}
        \T^{*}_{p}P=(\T_pP)^{*}\simeq (\Ver_p P\times \Hor_p P)^{*}\simeq (\Ver_p P)^{*}\times (\Hor_p P)^{*}\simeq\mathfrak{g}^{*}\times (\T_{\pi(p)}M)^{*}=\mathfrak{g}^{*}\times \T^{*}_{\pi(p)}M,
    \end{equation*}
    which leads us to a trivialisation of the cotangent bundle in the form of
    \begin{equation}
        \label{eq:triv_T*P}
        \imath_{\T^* P} : \T^* P \simeq P \times_M \T^* M \times \alg^*.
    \end{equation}

    \vspace{\spacesmall}

    The pairing of $(p, \varphi, A) \in P \times_M \T^* M \times \alg^*$ and $(p, \V, X) \in P \times_M \T M \times \alg$ takes the form
    \begin{align}
        \label{eq:ev_TP}
        \langle (p, \varphi, A), (p, \V, X) \rangle = \langle \varphi, \V \rangle + \langle A, X \rangle.
    \end{align}

    \vspace{\spacesmall}

    The trivialisations of $\T P$ and $\T^*P$ give rise to further trivialisations of the iterated tangent and cotangent bundles presented in the next subsections.

\subsection{Trivialisation of \textsf{TT}\textsl{P}}
    Using \eqref{eq:triv_TP}, we get the following isomorphisms (to keep calculations clear we denoted the second $ \T $ functor in red):
    \begin{equation*}
       \textcolor{red}{\T}\T P \simeq \textcolor{red}{\T}(P \times_M \T M \times \alg) \simeq \textcolor{red}{\T}P \times_{\textcolor{red}{\T}M} \textcolor{red}{\T}\T M \times \textcolor{red}{\T}\mathfrak{g} .
    \end{equation*}
    Now, we again apply \eqref{eq:triv_TP} to $\textcolor{red}{\T}P$. We also notice that since $\mathfrak{g}$ is a vector space, we have $\textcolor{red}{\T}\mathfrak{g}\simeq\mathfrak{g}\times\textcolor{red}{\mathfrak{g}}$ and
    \begin{equation*}
        \textcolor{red}{\T}\T P\simeq P\times_{M} \T M\times\textcolor{blue}{\mathfrak{g}}\times \T\T M\times\mathfrak{g}\times\textcolor{red}{\mathfrak{g}}.
    \end{equation*}
    Finally, we observe that the second factor does not introduce any new information about $\T\T P$ as it coincides with the image of the fourth factor under the map $\T\tau_{M}$. Thus, we have
    \begin{equation}
        \label{eq:triv_TTP}
        \imath_{\T\T P} : \T\T P \simeq P\times_{M} \T\T M\times\mathfrak{g}\times\textcolor{blue}{\mathfrak{g}}\times\textcolor{red}{\mathfrak{g}}.
    \end{equation}

    \begin{remark}
        This trivialisation of the double vector bundle $\T\T P$ is a double vector bundle isomorphism. The double vector bundle structure of the codomain may be presented on the diagram
        \begin{equation*}\begin{gathered}
            \scalebox{1}{\xymatrix@C-15pt@R-5pt{
                 & P\times_{M} \T\T M\times\mathfrak{g}\times\textcolor{blue}{\mathfrak{g}}\times\textcolor{red}{\mathfrak{g}} \ar[dl]|{(\pr_1, \tau_{\T M} \circ \pr_2, \pr_3)} \ar[dr]|{(\pr_1, \T \tau_M \circ \pr_2, \pr_4)} & \\
                P \times_M \T M \times \alg \ar[dr]|{\pr_1} & P \times_M \T M \times \textcolor{red}{\alg} \ar[d]^(0.5){\pr_1} \ar@{^{(}->}[u] & P \times_M \T M \times \textcolor{blue}{\alg} \ar[dl]|{\pr_1} \\
                 & P & .
            }}
        \end{gathered}\end{equation*}
        We note that all the incoming trivialised versions of bundles will possess similar double vector structures.
    \end{remark}

\subsection{Trivialisation of the canonical flip}

As an intermediate step in the construction of the Tulczyjew triple on the discussed geometric structure, we develop a trivialisation of the automorphism $\kappa_P$ of $\T\T P$. Later on, we will use it for building up the Tulczyjew isomorphism.


Let $\tau: E\rightarrow M$ be a vector bundle and $e,f\in E$ be elements of the same fibre of $E$, i.e. $\tau(e)=\tau(f)$. The vector tangent to the curve $t\mapsto e+tf$ will be denoted by $f^\Ver_e$ and called the {\it vertical lift of $f$ to $e$}. It is indeed a vertical vector at point the $e$.

\begin{proposition}\label{prop:1}
Let $X$ and $Y$ be two vector fields on a manifold $M$. The following formula holds:
\begin{equation*}
\T X(Y(x))-\kappa_M(\T Y(X(x)))=[Y,X]^{\Ver}_{X(x)}.
\end{equation*}
\end{proposition}

\begin{proof}
We will do the calculation in local, adapted coordinates. Starting from coordinates $(q^i)$ defined on some open subset $\mathcal{O}\subset M$, we construct, in the usual way, the coordinates $(q^i,\dot q^j)$ on $\tau_M^{-1}(\mathcal{O})$ and $(q^i, \dot q^j, \delta q^k, \delta \dot q^l)$ on $\tau_{\T M}^{-1}(\tau_M^{-1}(\mathcal{O}))$. If the vector fields $X$ and $Y$ are given in coordinates by functions $X^i(q)$ and $Y^i(q)$ respectively ($q$ denotes the whole set of coordinates $q^i$), then the map $\T X$ reads
$$\T X(q^i,\dot q^j)=\left(\,q^i,\, X^j(q), \,\dot q^k,\, \frac{\partial X^l}{\partial q^m}\dot q^m\,\right).$$
Similarly, $\T Y$ reads
$$\T Y(q^i,\dot q^j)=\left(\,q^i,\, Y^j(q),\, \dot q^k,\, \frac{\partial Y^l}{\partial q^m}\dot q^m\,\right).$$
Since $\kappa_M(q^i, \dot q^j, \delta q^k, \delta \dot q^l)=(q^i, \delta q^j, \dot q^k,\delta \dot q^l),$ the difference $\T X(Y)(x)-\kappa_M(\T Y(X(x)))$ reads
$$\T X(Y(x))-\kappa_M(\T Y(X(x)))=
\left(\,q^i,\, X^j(q),\, 0,\, \frac{\partial X^m}{\partial q^l}Y^l-\frac{\partial Y^m}{\partial q^l}X^l\,\right),$$
which is indeed the coordinate expression for $[X,Y]^{\Ver}_{X(x)}$.
\end{proof}

Using the above proposition, we may find the formula for $\tilde\kappa_P$, which is $\kappa_P$ composed with the trivialising $\imath_{\T\T P}$, i.e. the map
$$ \tilde\kappa_P: P\times_M\T\T M\times\mathfrak{g}\times\mathfrak{g}\times\mathfrak{g}\longrightarrow P\times_M\T\T M\times\mathfrak{g}\times\mathfrak{g}\times\mathfrak{g}.$$

First we will use Proposition \ref{prop:1} and two special vector fields on $P$. Let $\Vec{X}$ be the sum of the horizontal lift $\chi^\Hor$ of a vector field $\chi$ on $M$ with the fundamental vector field $\partial_X$ of an element $X\in\mathfrak{g}$. Such a field, in trivialisation, has the form
$$\imath_{\T P}\left(\vec{X}(p)\right)=(p, \chi(x), X) \in P \times_M \T M \times \alg,$$
where $x=\pi(p)$. The value of $\T\vec{X}$ on $(p,v,Y)$ reads in trivialisation
$$\T\vec{X}(p,v,Y)=(p,\T\chi(v), X,Y,0).$$
Therefore, we have
$$\T\vec{X}(\vec{Y}(p))=(p,\T\chi(\eta(x)),X,Y,0), \qquad
\T\vec{Y}(\vec{X}(p))=(p,\T\eta(\chi(x)),Y,X,0),$$
where $\vec{Y}=(p, \eta(q), Y)$ for a vector field $\eta$ on $M$ and $Y\in\mathfrak{g}$. The Lie bracket $[\vec{X},\vec{Y}]$ can be easily calculated. Indeed, we have
$$\vec{X}=\chi^\Hor+\partial_X, \qquad \vec{X}=\eta^\Hor+\partial_Y,$$
so also
\begin{align*}
[\vec{X},\vec{Y}]&=[\chi^\Hor+\partial_X,\eta^\Hor+\partial_Y ]=
[\chi^\Hor, \eta^\Hor]+ [\chi^\Hor,\partial_Y ]+[\partial_X,\eta^\Hor]+
[\partial_X, \partial_Y]\\
&=[\chi,\eta]^\Hor+\partial_{\Omega(\chi,\eta)}+\partial_{[X,Y]},
\end{align*}
where we used \eqref{eq:curv_abuse} and the fact that horizontal lifts are invariant vector fields, therefore $[\chi^\Hor,\partial_Y ]=0$ and $[\partial_X,\eta^\Hor]=0$.
In trivialisation, the bracket reads
$$
\imath_{\T P}\left([\vec{X},\vec{Y}](p)\right)=(\,p,\, [\chi,\eta](p),\, \Omega_p(\chi,\eta)+[X,Y]_{\mathfrak{g}}\, ) \in P \times_M \T M \times \alg.
$$
Using Proposition \ref{prop:1} for the vector fields $\vec{X}$ and $\vec{Y}$, we get in trivialisation
\begin{align*}
\tilde\kappa_P(p,\T\eta(\chi(q)), Y, X, 0)&=(\,p,\,\T\chi(\eta(q)),\, X,\, Y,\, 0\,) + (\,p,\, [\chi,\eta]^\Ver_{\chi(q)},\, X,\, 0,\,\Omega_p(\chi,\eta)+[X,Y]_{\mathfrak{g}}\,) \\
&=(\,p,\, \T\chi(\eta(q))+[\chi,\eta]^\Ver_{\chi(q)},\, X,\,Y,\, \Omega_p(\chi,\eta)+[X,Y]_{\mathfrak{g}}\,) \\
&=(\,p,\, \kappa_M(\T\eta(\chi(q)),\, X,\,Y,\, \Omega_p(\chi,\eta)+[X,Y]_{\mathfrak{g}}\,).
\end{align*}
Varying $\vec{X}$ and $\vec{Y}$, we conclude that
$$
\tilde\kappa_P(\, p,\, \mathcal{V},\, Y, X, 0)=
(\,p,\, \kappa_M(\mathcal{V}),\, X,\,Y,\, \Omega_p(\tau_{\sT M}(\mathcal{V}),\T\tau_M(\mathcal{V}))+[X,Y]_{\mathfrak{g}}\,).
$$
To determine the value of $\tilde\kappa_P$ on the most general element of the trivialised $\T\T P$, we recall that $\kappa_P$ is an automorphism of the double vector bundle $\T\T P$ over identities on both vector bundles, and it is an identity when restricted to the core. If $v$ is an element of the core of $\T\T P$ and $w$ is an element of $\sT\sT P$ over the same point in $M$, we have
\begin{equation}\label{eq:1}
\kappa_P(w+v)=\kappa_P(w)+v.
\end{equation}
In trivialisation, the core of $P\times_{M}\T\T P\times \mathfrak{g}\times\mathfrak{g}\times \mathfrak{g}$ consists of elements of the form $(p,u,0,0,Z)$, where $u$ is in the core of $\T\T M$. Using \eqref{eq:1}, we get the final formula for $\tilde\kappa_P$:
\begin{equation}
\label{eq:triv_kapp}
\tilde\kappa_P(\, p,\, \mathcal{V},\, Y, X, Z)=
(\,p,\, \kappa_M(\mathcal{V}),\, X,\,Y,\, Z+\Omega_p(\tau_{\sT M}(\mathcal{V}),\T\tau_M(\mathcal{V}))+[X,Y]_{\mathfrak{g}}\,).
\end{equation} 

\subsection{Trivialisation of \textsf{TT}*\textsl{P}}
    Trivialisation of the bundle $\T\T^* P$ is similar to that of $\T\T P$, except that we start with the trivialised version of the cotangent bundle, i.e. \eqref{eq:triv_T*P}, and then apply the tangent functor:
    \begin{equation*}
        \textcolor{black}{\T}\T^{*}P \simeq \textcolor{black}{\T}(P \times_M \T^* M \times \textcolor{red}{\mathfrak{g}^{*}}) \simeq \textcolor{black}{\T}P\times_{\textcolor{black}{\T}M}\textcolor{black}{\T}\T^{*}M\times\textcolor{black}{\T}\textcolor{red}{\mathfrak{g}^{*}}.
    \end{equation*}
    Now, as the dual of a Lie algebra is a vector space, we have $\textcolor{black}{\T}\textcolor{red}{\mathfrak{g}^{*}}\simeq \textcolor{red}{\mathfrak{g}^{*}}\times\textcolor{black}{\mathfrak{g}^{*}}$ and, by \eqref{eq:triv_TP}, we get (omitting the redundant $\T M$)
    \begin{equation*}
        \imath_{\T\T^* P}: \T\T^{*}P\simeq P\times_{M} \T\T^{*}M\times\textcolor{red}{\mathfrak{g}^{*}}\times\textcolor{blue}{\mathfrak{g}}\times\textcolor{black}{\mathfrak{g}^{*}}.
    \end{equation*}
    We coloured the Lie algebras and their duals to match the pairing discussed in the next subsection.

\subsection{Trivialisation of the pairing of \textsf{TT}*\textsl{P} and \textsf{TT}\textsl{P}}
    We now find the formula for the pairing between elements of trivialised $\T\T P$ and $\T\T^{*}P$. Let $v\in \T\T P$ and $\phi\in \T\T^{*}P$ be such that $\T\tau_{P}(v)=\T\pi_{P}(\phi)$. $v$ and $\phi$ may be written in the trivialisations as
    \begin{align*}
        \widetilde{v} &:= \imath_{\T\T P}(v)=(p,\V,X,Y,Z) \in P\times_{M} \T\T M\times\mathfrak{g}\times\textcolor{blue}{\mathfrak{g}}\times\textcolor{red}{\mathfrak{g}},\\
        \widetilde{\phi} &:= \imath_{\T\T^* P}(\phi)=(p,\varphi,A,Y,B) \in P\times_{M} \T\T^{*}M\times\textcolor{red}{\mathfrak{g}^{*}}\times\textcolor{blue}{\mathfrak{g}}\times\textcolor{black}{\mathfrak{g}^{*}}.
    \end{align*}
    Note that, since $\T\tau_{P}(v)=\T\pi_{P}(\phi)$, we have $\pr_1(\widetilde{v}) = \pr_1(\widetilde{\phi})$ and $\pr_4(\widetilde{v})=\pr_4(\widetilde{\phi})$. Moreover, $\T\tau_{M}(\V)=\T\pi_{M}(\varphi)$.

    Now, we pick a curve $\gamma_{\widetilde{v}}:\mathbb{R} \ni t \longmapsto \big(\gamma_p(t),\gamma_{\V}(t),X+tZ\big) \in \T P$ representing $\widetilde{v}$. Similarly, we pick a curve $\gamma_{\widetilde{\phi}}:\mathbb{R}\ni t \longmapsto \big(\gamma_p(t),\gamma_{\varphi}(t),A+tB\big) \in \T^{*}P$ representing $\widetilde{\phi}$ where $\gamma_{\V}(t)$ and $\gamma_{\varphi}(t)$ are curves representing $\V$ and $\varphi$ respectively, and $\gamma_p(t)$ is such a curve in $P$ that
    \begin{equation*}
        \imath_{\T P}(\der \, \gamma_p(t))=\big(p,\T\tau_{M}(\V),Y\big).
    \end{equation*}

    The evaluation of $\widetilde{v}$ on $\widetilde{\phi}$ reads
    \begin{align}
        \label{eq:triv_ev_TT*P}
        \langle\!\langle \widetilde{\phi},\widetilde{v}\rangle\!\rangle = \der \, \langle \gamma_{\widetilde{\phi}}, \gamma_{\widetilde{v}} \rangle = \der \,\big(\langle\gamma_{\varphi}(t),\gamma_{\V}(t)\rangle+\langle A+tB,X+tZ\rangle\big) \\
        \nonumber =\langle\!\langle\varphi,\V\rangle\!\rangle+\langle A,Z\rangle+\langle B,X\rangle,
    \end{align}
    where $\langle\!\langle\varphi,\V\rangle\!\rangle$ denotes the pairing between elements of the vector bundles $\T\T^{*}M$ and $\T\T M$.

\subsection{Trivialisation of \textsf{T}*\textsf{T}\textsl{P}}
    First, once again, we use \eqref{eq:triv_TP}:
    \begin{equation*}
        \T^{*}\T P\simeq \T^{*}(P\times_M \T M\times\mathfrak{g})\simeq \T^{*}(P\times_{M} \T M)\times \T^{*}\mathfrak{g}.
    \end{equation*}
    Now, we find a trivialisation of $\T^{*}(P\times_{M} \T M)$. Usually, there is no isomorphism $\T^{*}(P\times_{M} Q)\simeq \T^{*}P\times_{\T^{*}M} \T^{*}Q$, as $\T^*$ is not a covariant functor. Instead, for any fibre bundles $P$ and $Q$, we can write the isomorphism: $\T^{*}(P\times_{M} Q)\simeq (\T^{*}P\times_{M} \T^{*}Q)/(\T P\times_{\T M} \T Q)^{\circ}$. In our case, it takes the form
    \begin{align*}
        \T^*(P \times_M \T M) &\simeq (\T^* P \times_M \T^*\T M)/(\T P \times_{\T M} \T\T M)^\circ \\
        &\simeq (P \times_M \T^* M \times_M \T^*\T M \times \alg^*)/(P \times_M \times \T M \times_{\T M} \T\T M \times \alg)^\circ.
    \end{align*}
    We will find the form of this quotient in local coordinates. Let $(m^i)$ be the coordinates on some open $\mathcal{O} \subset M$, $(m^i,g^j)$ on $\pi^{-1}(\mathcal{O}) \subset P$ where $(g^j)$ are adopted to some open subset of the fibre $G$. Then $(m^i,\delta m^i)$ are the induced coordinates on $\tau_M^{-1}(M) \subset \T M$, $(m_i,p_i)$ on $\pi_M^{-1}(M) \subset \T^* M$, $(m^i,\dot{m}^i,\delta m^i,\delta\dot{m}^i)$ on $\tau_{\T M}^{-1} \circ \tau_M^{-1}(M)$ and $(m^i,\dot{m}^i,\varphi_i,\psi_i)$ on $\pi_{\T M}^{-1} \circ \tau_M^{-1}(M)$. Then, the covector
    \begin{align*}
        \tilde{\psi}=\left((m^i,g^j),(m_i,p_i), (m^i,\dot{m}^i,\varphi_i,\psi_i), A\right) \in P \times_M \T^* M \times_M \T^*\T M \times \alg^*
    \end{align*}
    belongs to the annihilator if $\langle\psi, v\rangle=0$ for every vector
    \begin{align*}
        \tilde{v}=\left((m^i,g^j), (m^i,\delta m^i), (m^i,\dot{m}^i,\delta m^i,\delta\dot{m}^i), X\right) \in P \times_M \T M \times_{\T M} \T\T M \times \alg.
    \end{align*}
    In trivialisation, the pairing is just component-wise, which in local coordinates gets us the equation
    \begin{equation*}
        \langle A,X\rangle+\delta\dot{m}^i\psi_{i}+\delta m^{i}(p_i+\varphi_i)=0,
    \end{equation*}
    which implies the conditions for $\psi$ being an element of the annihilator:
    \begin{equation*}
        A=0, \quad \psi_i=0, \quad p_i+\varphi_i=0.
    \end{equation*}
    Therefore, we see that the equivalence class of $\psi$ consists of elements having a constant sum $p_i+\varphi_i=c$. We can adopt a convenient convention that, for a given constant $c$, we will always choose the representative such that $p_i=0$ and $\varphi_i$=$c$. This allows us to write
    \begin{equation*}
        \T^{*}(P\times_{M} \T M)\simeq P\times_{M}\T^{*}\T M\times\mathfrak{g}^{*}.
    \end{equation*}
    Finally, using the fact that $\T^{*}\mathfrak{g}\simeq \mathfrak{g}\times\mathfrak{g}^{*}$, we obtain a trivialisation
    \begin{equation}
        \label{eq:triv_T*TP}
        \imath_{\T^*\T P}: \T^{*}\T P\simeq P\times_{M}\T^{*}\T M\times\mathfrak{g}\times\textcolor{blue}{\mathfrak{g}^{*}}\times\textcolor{red}{\mathfrak{g}^{*}}.
    \end{equation}

\subsection{Trivialisation of the pairing of \textsf{T}*\textsf{T}\textsl{P} and \textsf{TT}\textsl{P}}
    Any element $\psi \in \T^*\T P$ may be naturally evaluated on any $v \in \T\T P$ as long as $\pi_{\T P}(\psi) = \tau_{\T P}(v)$. For the trivialised counterparts
    \begin{align*}
        \widetilde{\psi} &:= \imath_{\T^*\T P}(\psi) = (p, \rho, X, A, B) \in P\times_{M} \T^{*}\T M\times\mathfrak{g}\times\textcolor{blue}{\mathfrak{g}^{*}}\times\textcolor{red}{\mathfrak{g}^{*}}, \\
        \widetilde{v} &:= \imath_{\T\T P}(v) = (p, \V, X, Y, Z) \in P\times_{M} \T\T M\times\mathfrak{g}\times\textcolor{blue}{\mathfrak{g}}\times\textcolor{red}{\mathfrak{g}},
    \end{align*}
    the evaluation condition takes the form of $\pr_1(\widetilde{\psi}) = \pr_1(\widetilde{v})$, $\pi_{\T M} \circ \pr_2(\widetilde{\psi}) = \tau_{\T M} \circ \pr_2(\widetilde{v})$ and $\pr_3(\widetilde{\psi}) = \pr_3(\widetilde{v})$. Using the findings of the previous subsection, the evaluation is given by
    \begin{align}
        \label{eq:triv_ev_T*TP}
        \langle \widetilde{\psi}, \widetilde{v} \rangle = \langle (p, \rho, X, A, B), (p, \V, X, Y, Z) \rangle = \langle\rho,\V\rangle+\langle A,Y\rangle+\langle B,Z\rangle.
    \end{align}

\subsection{Trivialisation of the Tulczyjew isomorphism}
    Now, we are ready to construct the trivialised Tulczyjew isomorphism $\widetilde{\alpha}_{P}$. We do it directly trivialising the definition $\langle\!\langle\phi,\kappa_P(v)\rangle\!\rangle=\langle\alpha_{P}(\phi),v\rangle$ to
    \begin{align}
        \label{eq:tulczyjew_iso_def}
        \langle\!\langle \widetilde{\phi},\widetilde{\kappa}_P(\widetilde{v})\rangle\!\rangle=\langle\widetilde{\alpha}_{P}(\widetilde{\phi}),\widetilde{v}\rangle
    \end{align} for any $v \in \T\T P$ and $\phi \in \T\T^* P$ having the trivialisations
    \begin{align*}
        \widetilde{v} &:= \imath_{\T\T P}(v)=(p,\V,Y,X,Z) \in P\times_{M} \T\T M\times\mathfrak{g}\times\textcolor{blue}{\mathfrak{g}}\times\textcolor{red}{\mathfrak{g}},\\
        \widetilde{\phi} &:= \imath_{\T\T^* P}(\psi)=(p,\varphi,A,Y,B) \in P\times_{M} \T\T^{*}M\times\textcolor{red}{\mathfrak{g}^{*}}\times\textcolor{blue}{\mathfrak{g}}\times\textcolor{black}{\mathfrak{g}^{*}}
    \end{align*}
    and satisfying the evaluation conditions $\T\pi_{M}(\varphi)=\T\tau_{M}(\kappa_{M}(\V))=\tau_{\T M}(\V)$.

    We start with the right-hand side of the defining equation \eqref{eq:tulczyjew_iso_def}. By \eqref{eq:triv_ev_T*TP}, we have
    \begin{align}
        \label{eq:tulczyjew_rhs}
        \langle\widetilde{\alpha}_{P}(\widetilde{\phi}),\widetilde{v}\rangle = \langle \pr_2(\widetilde{\alpha}_{P}(\widetilde{\phi})), \V \rangle + \langle \pr_4(\widetilde{\alpha}_{P}(\widetilde{\phi})), X \rangle + \langle \pr_5(\widetilde{\alpha}_{P}(\widetilde{\phi})), Z \rangle,
    \end{align}
    where in order for this evaluation to be well-defined $\pr_1(\widetilde{\alpha}_{P}(\widetilde{\phi})) = p$, $\pi_{\T M} \circ \pr_2(\widetilde{\alpha}_{P}(\widetilde{\phi})) = \tau_{\T M}(\V)$ and $\pr_3(\widetilde{\alpha}_{P}(\widetilde{\phi})) = Y$. To find the rest of the components, we will compare the above formula with the left-hand side of \eqref{eq:tulczyjew_iso_def}.

    Using \eqref{eq:triv_kapp}, \eqref{eq:triv_ev_TT*P} and $\T\pi_{M}(\varphi)=\tau_{\T M}(\V)$, it reads
    \begin{align*}
        \langle\!\langle \widetilde{\phi},\widetilde{\kappa}_P(\widetilde{v})\rangle\!\rangle&=\langle\!\langle(p,\varphi,A,Y,B),(p,\kappa_{M}(\V),X,Y,Z+\Omega_p(\tau_{\T M}\V,\T\tau_{M}\V)-[Y,X])\rangle\!\rangle\\
        &=\langle\alpha_{M}(\varphi),\V\rangle+
        \langle A,Z\rangle-\langle A,\ad_{Y}(X)\rangle+\langle A,\Omega_p(\T\pi_{M}(\varphi),\T\tau_M(\cdot))(\V)\rangle+\langle B,X\rangle.
    \end{align*}
    In order to further simplify the left-hand side, we use the dual of the linear mapping $\ad_Y : \alg \rightarrow \alg$, i.e. $\ad_Y^* : \alg^* \rightarrow \alg^*$ and the dual of the linear mapping $\Omega_p(\T\pi_{M}(\varphi),\T_u\tau_M(\cdot)): \T_u\T M \rightarrow \alg$, i.e. $\Omega_p(\T\pi_{M}(\varphi),\T_u\tau_M(\cdot))^*: \alg^* \rightarrow \T^*_u\T M$ where $u := \tau_{\T M}(\V) = \T\pi_{M}(\varphi)$. Because the latter is uniquely specified by just $(p, u) \in P \times_M \T M$, to ease the notation we will write $\Omega^*_{(p, \varphi)} \equiv \Omega^*_\varphi := \Omega_p(\T\pi_{M}(\varphi),\T_u\tau_M(\cdot))^*$, whenever there is only one projection onto $\T M$. Then,
    \begin{align*}
        \langle\!\langle \widetilde{\phi},\widetilde{\kappa}_P(\widetilde{v})\rangle\!\rangle&=\langle\alpha_{M}(\varphi),\V\rangle + \big\langle \Omega^*_\varphi(A),\V\big\rangle+\langle A,Z\rangle+\langle B-\ad_{Y}^{*}(A),X\rangle.
    \end{align*}
    We notice that $\Omega^*_\varphi(A)$ is an element of the core of the double vector bundle $\T^{*}\T M$ (the image of $(\T_u\tau_M)^*$ belongs to the core). Denoting addition in the core with $\dot{+}$, the left-hand side reads
    \begin{equation}
        \label{eq:tulczyjew_lhs}
        \langle\!\langle \widetilde{\phi},\widetilde{\kappa}_P(\widetilde{v})\rangle\!\rangle=\left\langle\alpha_{M}(\varphi)\dot{+}\Omega^*_\varphi(A),\V\right\rangle+ \langle B-\ad_{Y}^{*}(A),X\rangle + \langle A,Z\rangle.
    \end{equation}
    Comparing directly \eqref{eq:tulczyjew_rhs} and \eqref{eq:tulczyjew_lhs}, we obtain the rest of the components of $\widetilde{\alpha}_P$:
    \begin{align}
        \label{eq:triv_alpha}
        \widetilde{\alpha}_P &: P\times_{M} \T\T^{*}M\times\textcolor{red}{\mathfrak{g}^{*}}\times\textcolor{blue}{\mathfrak{g}}\times\textcolor{black}{\mathfrak{g}^{*}} \longrightarrow P\times_{M} \T^{*}\T M\times\mathfrak{g}\times\textcolor{blue}{\mathfrak{g}^{*}}\times\textcolor{red}{\mathfrak{g}^{*}} \\
        \nonumber &: (p,\varphi,A,Y,B) \longmapsto (p,\alpha_{M}(\varphi)\dot{+}\Omega^*_\varphi(A),Y,B-\ad_{Y}^{*}(A),A).
    \end{align} 

\subsection{Trivialisation of \textsf{T}*\textsf{T}*\textsl{P} and the pairing with \textsf{TT}*\textsl{P}}
    Trivialisation of $\T^*\T^* P$ is analogous to that of $\T^*\T P$. The only difference is that we start with the trivialisation of $\T^{*}P$ \eqref{eq:triv_T*P}, and then we repeat the reasoning which led to \eqref{eq:triv_T*TP}. We arrive at
    \begin{equation}
        \label{eq:triv_T*T*P}
        \imath_{\T^*\T^* P} : \T^{*}\T^{*}P\simeq P\times_{M}\T^{*}\T^{*}M\times\textcolor{red}{\mathfrak{g}^{*}}\times\textcolor{blue}{\mathfrak{g}^{*}}\times\textcolor{black}{\mathfrak{g}}.
    \end{equation}

    Similarly as for the bundles $\T^*\T P$ and $\T\T P$, there is a natural pairing of elements $\theta \in \T^*\T^* P$ and $\psi \in \T\T^* P$ as long as $\pi_{\T^* P}(\theta) = \tau_{\T^* P}(\psi)$. The trivialised counterparts read
    \begin{align*}
        \widetilde{\theta} &:= \imath_{\T^*\T^* P}(\phi) = (p, f, A, B, X) \in P\times_{M} \T^{*}\T M\times\textcolor{red}{\mathfrak{g}^*}\times\textcolor{blue}{\mathfrak{g}^{*}}\times\textcolor{black}{\mathfrak{g}}, \\
        \widetilde{\psi} &:= \imath_{\T\T^* P}(\psi) = (p, \varphi, A, Y, C) \in P\times_{M} \T\T M\times\textcolor{red}{\mathfrak{g}^*}\times\textcolor{blue}{\mathfrak{g}}\times\textcolor{black}{\mathfrak{g}^*},
    \end{align*}
    and the evaluation condition takes the form $\pr_1(\widetilde{\theta}) = \pr_1(\widetilde{\psi})$, $\pi_{\T^* M} \circ \pr_2(\widetilde{\theta}) = \tau_{\T^* M} \circ \pr_2(\widetilde{\psi})$ and $\pr_3(\widetilde{\theta}) = \pr_3(\widetilde{\psi})$. The evaluation is given by
    \begin{align}
        \label{eq:triv_pairing_T*T*P}
        \langle \widetilde{\theta}, \widetilde{\psi} \rangle = \langle (p, f, A, B, X), (p, \varphi, A, Y, C) \rangle = \langle f,\varphi \rangle + \langle B,Y\rangle + \langle C,X\rangle.
    \end{align}

\subsection{Trivialisation of the canonical isomorphism of \textsf{T}*\textsf{T}*\textsl{P} and \textsf{T}*\textsf{T}\textsl{P}}
    In \cite{canonical_gamma}, Dufour introduced the canonical isomorphism $\gamma_{E}:\T^* E^*\rightarrow \T^* E$ for a vector bundle $E$. The graph of this morphism is precisely the Lagrangian submanifold $\mathcal{L}_{\epsilon}$ generated by the pairing function $\epsilon$ between $E^*$ and $E$, defined on the submanifold $E^\ast\times_M E\subset E^\ast\times E$.

    We will now repeat this construction but taking the trivialised version of $\T P$ as the vector bundle and \eqref{eq:ev_TP} as the generating function ${\widetilde{\epsilon}}$. $\mathcal{L}_{\widetilde{\epsilon}}$ is now a submanifold of a trivialisation of $\T^{*}(\T^* P\times \T P)$, i.e. $(P\times_{M}\T^{*}\T^* M\times\alg^{*}\times\textcolor{blue}{\alg^{*}}\times\textcolor{red}{\alg})\times (P\times_{M}\T^{*}\T M\times\alg\times\textcolor{blue}{\alg^{*}}\times\textcolor{red}{\alg^{*}})$.
    We want to find the conditions that a general covector $\Theta = (p,\varphi,B,A,X; p,\psi,Y,C,D)$, projecting onto the trivialised $\T^*P \times_P \T P$, needs to meet in order to be in $\mathcal{L}_{\widetilde{\epsilon}}$. For this, we use the equation induced by the generating function
    \begin{align}
        \label{eq:generating_of_submanifold}
        \langle \Theta , \mathcal{U} \rangle = \dif \widetilde{\epsilon} (\mathcal{U}),
    \end{align}
    where $\mathcal{U}$ is any element of a trivialisation of $\T(\T^*P \times_P \T P) \simeq \T\T^* P \times_{\T P} \T\T P$ that may be evaluated on $\Theta$:
    \begin{equation*}
        \mathcal{U} = (p,f,B,Z,\delta B; p,u,Y,Z,\delta Y) \in (P\times_{M} \T\T^* M\times\alg^{*}\times\textcolor{blue}{\alg}\times\textcolor{red}{\alg^{*}}) \times_{\T P} (P\times_{M} \T\T M\times\alg\times\textcolor{blue}{\alg}\times\textcolor{red}{\alg}).
    \end{equation*}
    By \eqref{eq:triv_ev_T*TP} and \eqref{eq:triv_pairing_T*T*P}, the evaluation on the left-hand side of \eqref{eq:generating_of_submanifold} reads
    \begin{align*}
        \langle \Theta , \mathcal{U} \rangle = \langle A,Z\rangle+\langle\delta B,X\rangle+\langle C,Z\rangle+\langle D, \delta Y\rangle+\langle\varphi, f\rangle+\langle\psi,u\rangle.
    \end{align*}
    We identify the right-hand side as the evaluation \eqref{eq:triv_ev_TT*P} understood as a tangent mapping:
    \begin{equation*}
        \dif \widetilde{\epsilon} (\mathcal{U}) = \llangle f, u \rrangle + \langle \delta B, Y \rangle + \langle B, \delta Y \rangle.
    \end{equation*}
    Comparing both sides, we get
    \begin{align*}
        D = B, \quad Y = X, \quad A + C = 0 \quad \textrm{and} \quad \llangle f, u \rrangle = \langle\varphi, f\rangle+\langle\psi,u\rangle.
    \end{align*}
    The last condition encodes the canonical isomorphism $\gamma_{\T M} : \T^*\T^* M \rightarrow \T^*\T M$. Therefore, elements of the submanifold $\mathcal{L}_{\widetilde{\epsilon}}$ take the form
    \begin{equation*}
        \left ( p, \varphi, B, A, X; p, \gamma_{\T M}(\varphi), X, -A, B \right ).
    \end{equation*}
    As this submanifold is the graph of $\widetilde{\gamma}_{\T P}$, we get
    \begin{align}
        \label{eq:triv_canonical_iso}
        \widetilde{\gamma}_{\T P} &: P\times_{M}\T^{*}\T^* M\times\alg^{*}\times\textcolor{blue}{\alg^{*}}\times\textcolor{red}{\alg} \longrightarrow P\times_{M}\T^{*}\T M\times\alg\times\textcolor{blue}{\alg^{*}}\times\textcolor{red}{\alg^{*}} \\
        \nonumber &:( p, \varphi, B, A, X) \longmapsto \left( p, \gamma_{\T M}(\varphi), X, -A, B \right ).
    \end{align}

\subsection{Trivialisation of the symplectic-based isomorphism}
    As mentioned earlier, the symplectic-based isomorphism may be written as $\beta_M=\gamma_{\T M}^{-1}\circ\alpha_{M} : \T\T^* M \rightarrow \T^*\T^* M$. Therefore, in the case of a trivialised principal bundle with connection, the formula reads
    \begin{align*}
        \widetilde{\beta}_P(p,\varphi,A,Y,B) &= \widetilde{\gamma}_{\T P}^{-1}\circ \widetilde{\alpha}_{P}(p,\varphi,A,Y,B)\\
        &=\widetilde{\gamma}_{\T P}^{-1}(p,\alpha_{M}(\varphi)\dot{+}\Omega^*_\varphi(A),Y,B-\ad_{Y}^{*}(A),A)\\
        &=\left(p,\gamma_{\T M}^{-1}\big(\alpha_{M}(\phi)\dot{+}\Omega^*_\varphi(A)\big),A,\ad_{Y}^{*}(A)-B,Y\right).
    \end{align*}
    We may, once again, rewrite it using the fact that both double vector bundles $\T^*\T^* M$ and $\T^* \T M$ have the same core, and $\gamma_{\T M}$ inverts the sign of vectors in the core. Thus, denoting by $\dot{-}$ subtraction in the core of $\T^{*}\T^{*}M$, we arrive at the final form of the trivialised symplectic-based isomorphism:
    \begin{align}
        \label{eq:triv_beta}
        \widetilde{\beta}_P &: P\times_{M} \T\T M\times\textcolor{red}{\mathfrak{g}^*}\times\textcolor{blue}{\mathfrak{g}}\times\textcolor{black}{\mathfrak{g}^*} \longrightarrow P\times_{M}\T^{*}\T^* M\times\alg^{*}\times\textcolor{blue}{\alg^{*}}\times\textcolor{red}{\alg} \\
        \nonumber &: (p,\varphi,A,Y,B) \longmapsto \left(p,\beta_{M}(\varphi)\dot{-}\Omega^*_\varphi(A),A,\ad_{Y}^{*}(A)-B,Y\right).
    \end{align} 

\subsection{Dynamics on a principal bundle with connection}
    The trivialised Tulczyjew triple on a principal bundle with connection is the following diagram summarising the above results:
    \begin{equation}\begin{gathered}
        \label{xy:triple_on_P}
        \scalebox{1}{\xymatrix@C-15pt@R-10pt{
            P \times_M \T^*\T^* M \times \alg^* \times \alg^* \times \alg \ar[dr]|{(\pr_1, \pi_{\T^* M}\circ \pr_2, \pr_3)} & & \\
             & P \times_M \T^* M \times \alg^* \ar[dr]|{\pr_1} & \\
            P \times_M \T\T^* M \times \alg^* \times \alg \times \alg^* \ar[uu]_{\tilde{\beta}_P} \ar[dd]^{\tilde{\alpha}_P} \ar[ur]|{(\pr_1, \tau_{\T^* M} \circ \pr_2, \pr_3)} \ar[dr]|{(\pr_1, \T\pi_M \circ \pr_2, \pr_4)} & & P \\
             & P \times_M \T M \times \alg \ar[ur]|{\pr_1} & \\
            P \times_M \T^*\T M \times \alg \times \alg^* \times \alg^* \ar[ur]|{(\pr_1, \pi_{\T M} \circ \pr_2, \pr_3)} & & .
         }}
    \end{gathered}\end{equation}

    We may consider a mechanical system with a principal bundle serving as the configuration manifold and a Lagrangian $L : \T P \rightarrow \R$. We may introduce the trivialised Lagrangian $\widetilde{L} := L \circ \imath_{\T P}^{-1} : P \times_M \T M \times \alg \rightarrow \R$. The trivialised dynamics of the system $\widetilde{\mathcal{D}} = \widetilde{\alpha}_P^{-1} \circ \dif \widetilde{L} (P \times_M \T M \times \alg) \in P \times_M \T\T^* M \times \alg^* \times \alg \times \alg^*$ is given by
    \begin{gather}
        \label{eq:dynamics}
        \dif \widetilde{L} (V) \equiv \dif \widetilde{L} (p, v, X) = (p, \partial_{\T M} \widetilde{L}_V, X, \partial_G \widetilde{L}_V, \partial_\alg \widetilde{L}_V) \in P \times_M \T^*\T M \times \alg \times \alg^* \times \alg^*, \\
        \nonumber \widetilde{\mathcal{D}} = \left \{ \left. \begin{gathered} (p, \varphi, A, X, B) \in \\ P \times_M \T\T^* M \times \alg^* \times \alg \times \alg^*\end{gathered} \right| \begin{aligned}\varphi &= \alpha_M^{-1}(\partial_{\T M} \widetilde{L}_V) \dot{-} \Omega^*_{(p, \partial_{\T M}\widetilde{L}_V)}(\partial_\alg \widetilde{L}_V), \\
        A &= \partial_\alg \widetilde{L}_V, \quad B = \partial_G \widetilde{L}_V + \ad_X^*(\partial_\alg \widetilde{L}_V) \\ &\text{for some } V \in P\times_M \T M \times \alg \end{aligned} \right \},
    \end{gather}
    where $\partial_{\T M} \widetilde{L}_V \in \T^*\T M$ we identify with the restriction of $\dif \widetilde{L}(V)$ to the subspace $\T_v\T M$, i.e. $\partial_{\T M} \widetilde{L}_V := \dif \widetilde{L}(V) \circ \pr_{\T_v\T M}$ ($\pr_{\T_v\T M}$ is the projection onto the first component of the direct sum (the fibre) $\T_v\T M \oplus \alg \oplus \alg$). Similarly, $\partial_G \widetilde{L}_V := \dif \widetilde{L}(V) \circ \pr_{\alg_1}$, which may be written as $\partial_G \widetilde{L}_V = \der\, \widetilde{L}(p\exp(tX_a), v, X) A^a$, where $\{X_a\}$ is a basis of $\alg$ and $\{A^a\}$ is its dual. Lastly, $\partial_\alg \widetilde{L}_V := \dif \widetilde{L}(V) \circ \pr_{\alg_2}$, which may be written as $\partial_G \widetilde{L}_V = \der\, \widetilde{L}(p, v, X + tX_a) A^a$.

    Now, let $H : \T^* P \rightarrow \R$ be the Hamiltonian of a system with the trivialisation $\widetilde{H} := H \circ \imath_{\T^* P}^{-1} : P \times_M \T^* M \times \alg^* \rightarrow \R$. The dynamics of the system, i.e. the image of the Hamiltonian vector field $\widetilde{X}_H := \widetilde{\beta}_P^{-1} \circ \dif \widetilde{H}$, takes the form
    \begin{gather}
        \label{eq:dynamics_ham}
        \dif \widetilde{H}(\mathfrak{A}) \equiv \dif \widetilde{H} (p, \alpha, A) = (p, \partial_{\T^* M} \widetilde{H}_\mathfrak{A}, A, \partial_G \widetilde{H}_\mathfrak{A}, \partial_{\alg^*} \widetilde{H}_\mathfrak{A}) \in P \times_M \T^*\T^* M \times \alg^* \times \alg^* \times \alg, \\
        \nonumber \widetilde{\mathcal{D}} = \left \{ \left. \begin{gathered} (p, \varphi, A, X, B) \in \\ P \times_M \T\T^* M \times \alg^* \times \alg \times \alg^*\end{gathered} \right| \begin{aligned}\varphi &= \beta_M^{-1}(\partial_{\T^* M} \widetilde{H}_\mathfrak{A}) \dot{-} \Omega^*_{(p, \partial_{\T^* M} \widetilde{H}_\mathfrak{A})}(A), \\
        X &= \partial_{\alg^*} \widetilde{H}_\mathfrak{A}, \quad B = \partial_G \widetilde{H}_\mathfrak{A} + \ad_{\partial_{\alg^*} \widetilde{H}_\mathfrak{A}}^*(A) \\ &\text{for some } \mathfrak{A} \in P\times_M \T^* M \times \alg^* \end{aligned} \right \},
    \end{gather}
    where $\partial_{\T^* M} \widetilde{H}_\mathfrak{A}$, $\partial_G \widetilde{H}_\mathfrak{A}$ and $\partial_{\alg^*} \widetilde{H}_\mathfrak{A}$ are defined in a similar manner to the Lagrangian case.

\subsection{Example: Tulczyjew triple on the frame bundle of a sphere}
    \label{sec:example1}
    In this chapter, we present the machinery derived in the previous subsections in action. We start by considering a $2$-dimensional sphere. Let $\fr$ denote the principal bundle of oriented orthogonal frames. We will now describe a more convenient form of presenting $\fr$.

    First, we embed $S^2$ into $\R^{3}$ as a unit sphere centred at the origin. $\R^{3}$ is an affine space, hence every tangent vector $v \in \T S^2$ may be treated as a vector in $\R^{3}$ (which we denote with $\vec{v}$).

    Furthermore, every point of $S^2$ may be represented by a unit vector $\vec{n}$ from the origin to the said point.  Therefore, every point of $\fr$ may be represented by three orthonormal vectors in $\R^{3}$ (we add $\vec{n}$ as the $0$th vector of the frame).

    In this presentation, some other constructions also take convenient forms.
    \begin{itemize}
        \item There is a (left) $SO(3)$-action on $\fr$ of the form
            \begin{align*}
                SO(3) \times \fr \ni \big(\chi, p \big) \longmapsto \chi \cdot p = \left(\chi\vec{n}, \big(\chi\vec{f}_1, \chi\vec{f}_2\big)\right) \in \fr,
            \end{align*}
            where $f := (\vec{f}_1, \vec{f}_2)$ denotes an element of the space $\T_{\vec{n}}S^2$. As the action is transitive and free, $\fr$ (identified with $3$ vectors in $\R^{3}$) is a principal homogeneous $SO(3)$-space. What follows is that the (left) fundamental vector fields corresponding to this action (in the antihomomorphism convention)
            \begin{align*}
                \fr \times so(3) \ni (p, X) \longmapsto \der \, \exp(tX) \cdot p = Xp
            \end{align*}
            span the tangent space at every point, and $\T \fr \simeq \fr \times so(3)$.
        \item The algebra $so(3)$ is isomorphic to a Lie algebra built out of the vector space $\R^3$ with a bracket given by the standard cross product. The isomorphism reads
            \begin{align*}
                so(3) \; : \;
                X := \begin{bmatrix}
                    0 & -z & y \\
                    z & 0 & -x \\
                    -y & x & 0
                \end{bmatrix} \longmapsto
                \begin{bmatrix}
                    x \\
                    y \\
                    z
                \end{bmatrix}
                =: \vec{x} \in \R^3.
            \end{align*}
            What is more, the isomorphism turns the action of the standard representation of $so(3)$ on $\R^3$ into the cross product: $X\vec{z} = \vec{x} \times \vec{z}$. Additionally, $\vec{x}$ is the axis of the rotation generated by $X \in so(3)$, where the angle of the rotation is given by $t\|\vec{x}\|$ (with the standard Euclidean scalar product). A tedious calculation shows also that, for $\chi \in SO(3)$, the matrix representing the vector $\chi \vec{n}$ is precisely $\AdG_\chi (N)$, where $\vec{n} \simeq N$.
        \item The structure (right) $SO(2)$-action takes the form
            \begin{align*}
                \fr \times SO(2) \ni (p, \lambda) \longrightarrow p \triangleleft \lambda := \lambda^{-1}_{\vec{n}} \cdot p = (\vec{n}, \lambda^{-1}_{\vec{n}}(f)) \in \fr,
            \end{align*}
            where $\lambda_{\vec{n}}$ denotes the embedding of $SO(2)$ into $SO(3)$ as a subgroup of rotations around the axis $\vec{n}$. We obviously identify the trivial algebra $so(2)$ with $\R$, and the fundamental vector fields of this action may be written as
            \begin{equation*}
                \fr \times \R \ni (p, x) \longmapsto \der \, p \triangleleft \exp(txN) = -xNp.
            \end{equation*}
        \item We now also have a convenient way of thinking about curves representing a vector $(\vec{n}, \vec{v}) \in \T S^2$ (where $\vec{n}$ and $\vec{v}$ are orthogonal) as $\exp(tX)\vec{n}$ such that $X\vec{n} = \vec{v}$. Furthermore, any element of $(\vec{n}, \vec{v}, \vec{w}, \vec{u}) \in \T\T S^2$ may be represented by a curve of the form $(\exp(tY)\vec{n}, (1 + yt)\exp(tY)\exp(txN)\vec{v}) \in \T S^2$ or, equivalently, a homotopy $h(s, t) = \exp(sY)\exp(sxN)\exp(t(1+ys)X)\vec{n}$ where $X\vec{n} = \vec{v}$, $Yn = \vec{w}$ and $(Y+xN+y)\vec{v} = \vec{u}$. This leads to the canonical flip on the sphere in the form of
            \begin{align*}
                \kappa_{S^2}(\vec{n}, \vec{v}, \vec{w}, \vec{u}) = (\vec{n}, \vec{w}, \vec{v}, \vec{u}).
            \end{align*}
        \item It is worth noting that the actions of $SO(2)$ and $SO(3)$ commute, in the sense that
            \begin{align}
                \label{eq:SO2_SO3_comm}
                \chi \circ \lambda_{\vec{z}} = \lambda_{\chi(\vec{z})} \circ \chi.
            \end{align}
        \item Using \eqref{eq:SO2_SO3_comm}, we complete the picture with the formula for the tangent action of $\lambda \in SO(2)$ on $(p, X) \in \fr \times so(3)$:
            \begin{align}
                \label{eq:tangent_action_SO2}
                (p, X) \triangleleft \lambda &= \der \, \lambda_{\exp(tX)p}^{-1} \exp(tX) p = \der \, \exp(tX) \lambda_p^{-1} p \equiv (\lambda_p^{-1} p, X).
            \end{align}
    \end{itemize}

    The above constructions provide a way to transport a frame between fibres with the help of an element of $SO(3)$, but the transport is not unique. Now, we will specialise it to a proper connection.

    \vspace{5mm}
    We recall that the vertical subbundle is defined as the kernel of $\T\pi_\fr$. Hence, we first calculate the projection onto $\T S^2$ of the trivialised tangent vector $((\vec{n}, f), X) \in \fr \times so(3) \simeq \T \fr$:
    \begin{align*}
        \T \pi_\fr((\vec{n}, f), X) &= \der \, \pi_\fr \big(\exp(tX)\vec{n}, \exp(tX)f\big) = X\vec{n} \equiv \vec{x} \times \vec{n}.
    \end{align*}
    Therefore, the vertical bundle over every $\vec{n} \in S^2$ consists of vectors satisfying $\vec{n} \times \vec{x} = 0$. Using the Euclidean metric on $\R^3$ we may express this condition in an alternative form $\langle \vec{n} \mid \vec{x} \rangle = \| \vec{x} \|$, which yields
    \begin{align*}
        \Ver \fr = \left\{ \big( (\vec{n}, f), \vec{x} \big) \in \fr \times so(3) \; \mid \; \langle \vec{n} \mid \vec{x} \rangle = \| \vec{x} \| \right\}.
    \end{align*}
    This inspires a connection on $\fr$ in the form of the orthogonal complement $\Hor_p\fr := (\Ver_p\fr)^\bot$. The projections onto the vertical and horizontal subbundles read
    \begin{align*}
        \pr_\Ver &: \T \fr \longrightarrow \Ver \fr : ((\vec{n}, f), \vec{x}) \longmapsto ((\vec{n}, f), \langle \vec{n} \mid \vec{x} \rangle \vec{n}), \\
        \pr_\Hor &: \T \fr \longrightarrow \Hor \fr : ((\vec{n}, f), \vec{x}) \longmapsto ((\vec{n}, f), \vec{x} - \langle \vec{n} \mid \vec{x} \rangle \vec{n}).
    \end{align*}
    The induced connection one-form takes the form
    \begin{align*}
        \omega_\fr &: \T\fr \longrightarrow so(2) \simeq \R : (p, X) \equiv ((\vec{n}, f), \vec{x}) \longmapsto \langle \vec{n} \mid \vec{x} \rangle,
    \end{align*}
    and the principality of the connection is ensured by the calculation
    \begin{align*}
            (\wp_\lambda^* \omega_\fr)((\vec{n}, f), \vec{x}) &= \omega_\fr\big(\T \wp_\lambda((\vec{n}, f), \vec{x})\big) \equiv \omega_\fr \big( \lambda^{-1}_{\vec{n}}(\vec{n}, f), \vec{x} \big) \\
            &= \omega_\fr \big( (\vec{n}, \lambda^{-1}_{\vec{n}}f), \vec{x} \big) = \langle \vec{n} \mid \vec{x} \rangle = \AdG_{\lambda^{-1}}(\langle \vec{n} \mid \vec{x} \rangle).
    \end{align*}
    We consider $\vec{v}\in \T_qS^2$. The horizontal lift takes the form $\vec{v}_p^\Hor=(p,\vec{n}\times\vec{v})$. Indeed, it is horizontal as $\omega(p,\vec{n}\times\vec{v}) = \bra{\vec{n}}\ket{\vec{n}\times\vec{v}}=0$ and $\T\pi_\fr(p,\vec{n}\times\vec{v}) = (\vec{n}\times\vec{v}) \times \vec{n} =\vec{v}$.

    \vspace{\spacemid}

    In order to formulate the trivialised Tulczyjew triple on the frame bundle, we have to find the curvature two-form.

    We will use the identity \eqref{alternative_omega}. We consider two fundamental vector fields -- generated by $X$ and $Y$. First, we want to calculate the Lie bracket of their horizontal parts (we treat $\vec{n}$ as the tautological vertical vector field on $\fr$, which assigns to a point on the sphere its unit vector: $\fr \ni p \longmapsto \pi_\fr(p) \simeq \vec{n} \simeq N \in so(3)$):
    \begin{equation*}
        [X-\bra{\vec{n}}\ket{\vec{x}}N,Y-\bra{\vec{n}}\ket{\vec{y}}N]=[X,Y]-[\bra{\vec{n}}\ket{\vec{x}}N,Y]-[X,\bra{\vec{n}}\ket{\vec{y}}N]+[\bra{\vec{n}}\ket{\vec{x}}N,\bra{\vec{n}}\ket{\vec{y}}N].
    \end{equation*}
    The components above become:
    \begin{itemize}
        \item $[X, Y] = [Y, X]_{so(3)} = \vec{y} \times \vec{x}$,
        \item $\vec{n}$ is a vector field conserved by any rotation of $S^2$. Therefore, $\vec{n}$ commutes with the generators of the rotations, i.e. $[N,X]=[N,Y]=0$. This enables the calculation
            \begin{align*}
                [\bra{\vec{n}}\ket{\vec{x}}N,Y]&=\bra{\vec{n}}\ket{\vec{x}}[N,Y]-Y(\bra{\vec{n}}\ket{\vec{x}})N=-Y(\bra{\vec{n}}\ket{\vec{x}})N=-\der \,\bra{\exp(tY)\vec{n}}\ket{\vec{x}}N \\
                &=-\der \,\bra{\vec{n}}\ket{\exp(-tY)\vec{x}}N=\bra{\vec{n}}\ket{\vec{y}\times\vec{x}}N,
            \end{align*}
        \item we get the same result for
            \begin{equation*}
                [X,\bra{\vec{n}}\ket{\vec{y}}N]=\bra{\vec{n}}\ket{\vec{y}\times\vec{x}}N,
            \end{equation*}
        \item for the last component, we use the product rule twice:
            \begin{equation*}
                [\bra{\vec{n}}\ket{\vec{x}}N,\bra{\vec{n}}\ket{\vec{y}}N]=\bra{\vec{n}}\ket{\vec{x}}\bra{\vec{n}}\ket{\vec{y}}[N,N]+\big(\bra{\vec{n}}\ket{\vec{x}}N\bra{\vec{n}}\ket{\vec{y}}-\bra{\vec{n}}\ket{\vec{y}}N\bra{\vec{n}}\ket{\vec{x}}\big)N=0.
            \end{equation*}
    \end{itemize}
    Finally, the curvature two-form takes the form
    \begin{equation*}
        \Omega(X, Y) = -\omega([X-\bra{\vec{n}}\ket{\vec{x}}N,Y-\bra{\vec{n}}\ket{\vec{y}}N]) = -\bra{\vec{n}}\ket{\vec{y}\times\vec{x}}+\bra{\vec{n}}\ket{\vec{y}\times\vec{x}}+\bra{\vec{n}}\ket{\vec{y}\times\vec{x}} = \bra{\vec{n}}\ket{\vec{y}\times\vec{x}}.
    \end{equation*}

    Using the calculated horizontal lift, we can express the curvature two-form in the way present in the Tulczyjew triple. We take $v,w\in T_{\vec{n}}S^2$. Then,
    \begin{equation*}
        \Omega_p(\vec{v},\vec{w}) \equiv \Omega(\vec{v}_p^\Hor,\vec{w}_p^\Hor)=\bra{\vec{n}}\ket{(\vec{n}\times\vec{w})\times(\vec{n}\times\vec{v})}=\bra{\vec{n}}\ket{\bra{\vec{n}}\ket{\vec{w}\times\vec{v}}\vec{n}}=\bra{\vec{n}}\ket{\vec{w}\times\vec{v}}.
    \end{equation*}

    \vspace{5mm}

    Having prepared all the necessary components, we can trivialise the iterated tangent bundle and, in consequence, the canonical flip. The choice of the connection provides the explicit form of the trivialisation map $\imath_{\fr}$ in the case of the frame bundle identified with three vectors in $\R^3$:
    \begin{align*}
        \imath_{\T\fr} &: \T \fr \simeq \fr \times_{S^2} \T S^2 \times so(2) \\
        &: \big((\vec{n}, f), \vec{x} \big) \longmapsto ((\vec{n}, f), \vec{x} \times \vec{n}, \langle \vec{n} \mid \vec{x} \rangle ).
    \end{align*}
    By \eqref{eq:triv_TTP} and an identification $so(2) \simeq \R$, trivialised $\T\T \fr$ becomes
    \begin{align} \label{eq:ex_TTP}
        \imath_{\T\T\fr} : \T \T \fr \simeq \fr \times_{S^2} \T\T S^2 \times \R \times \R \times \R.
    \end{align}

    To formulate the first canonical isomorphism, we consider
    \begin{gather*}
        p = (\vec{n}, f) \in \fr, \quad \V = (\vec{n}, \vec{v}, \vec{w}, \vec{u}) \in \T\T S^2, \quad x, y, z \in \R \simeq so(2).
    \end{gather*}
    Using Formula \eqref{eq:triv_kapp}, we get the following form of the canonical flip on the trivialised frame bundle:
    \begin{gather*}
        \tilde{\kappa}_\fr\big(p, (\vec{n}, \vec{v}, \vec{w}, \vec{u}), x, y, z\big) = \big(p, (\vec{n}, \vec{w}, \vec{v}, \vec{u}), y, x, z + \bra{\vec{n}}\ket{\vec{w}\times\vec{v}}\big).
    \end{gather*}
    This formula is a little simpler than the general one because $so(2)$ is a commutative algebra.

    Due to the form of $so(2)$, we also have $so^*(2) \simeq \R^* \simeq \R$, and trivialisations of the rest of iterated tangent and cotangent bundles of the principal bundle are very similar to \eqref{eq:ex_TTP}.

    We find the form of $\Omega(\vec{v}_p^\Hor,(\pr_3(\cdot))^\Hor_p)^{*}(a)$ for $a \in so(2)^{*}\cong\mathbb{R}$ using duality:
    \begin{equation*}
        \big\langle\Omega(\vec{v}_p^\Hor,\pr_3(\cdot)^\Hor_p)^{*}(a),\V\big\rangle=\langle a,\Omega(\vec{v}_p^\Hor,\pr_3(\V)^\Hor_p)\rangle = a\bra{\vec{v}\times\vec{n}}\ket{\pr_3(\V)}.
    \end{equation*}
    Therefore,
    \begin{equation*}
        \Omega_p(\vec{v},\pr_3(\cdot))^{*}(a)=a\langle\vec{v}\times\vec{n}\mid\pr_3(\cdot)\rangle \equiv (\vec{n}, 0, a\langle\vec{v}\times\vec{n}\mid\cdot \rangle, 0) \simeq (\vec{n}, 0, a\vec{v}\times\vec{n}, 0) \in \T^*\T S^2,
    \end{equation*}
    where we identified the spaces  $(\R^3)^* \simeq \R^3$ via the dot product. To calculate the trivialised Tulczyjew triple on the frame bundle of $S^2$, we consider
    \begin{gather*}
        \Phi \in \T\T^* S^2, \quad \T\pi_{S^2}(\Phi) = (\vec{n}, \vec{v}), \quad a, b \in \R \simeq so(2)^*.
    \end{gather*}
    By \eqref{eq:triv_alpha}, the Tulczyjew isomorphism reads
    \begin{gather*}
        \tilde{\alpha}_\fr(p, \Phi, a, y, b) = (p, \alpha_{S^2}(\Phi) \dot{+} a \langle \vec{v} \times \vec{n} \mid \pr_3(\cdot) \, \rangle, y, b, a).
    \end{gather*}
    By \eqref{eq:triv_beta}, the symplectic-based isomorphism reads
    \begin{gather*}
        \tilde{\beta}_\fr(p, \Phi, a, y, b) = (p, \beta_{S^2}(\Phi) \dot{-} a \langle \vec{v} \times \vec{n} \mid \pr_3(\cdot) \, \rangle, -b, a, y).
    \end{gather*}
    Once again, we note that $\ad^*_x = 0$ simplifies the formulae.

\subsection{Example: Dynamics of a system modelled by the frame bundle on the sphere}
    A natural physical example of a system with configuration space modelled by the frame bundle is an axially symmetric rigid body confined to move on a sphere in a way that the symmetry axis is always normal to the sphere. The Lagrangian of such a system, explicitly exercising the connection on the frame bundle, is
    \begin{align*}
        \imath_{\T\fr}(\fr) \ni (p, (\vec{n}, \vec{v}), r) \longmapsto L(p, (\vec{n}, \vec{v}), r) := \frac12 I_\perp \|\vec{n} \times \vec{v}\|^2 + \frac12 I_{ax}r^2 \in \R,
    \end{align*}
    where $I_\perp$ is the moment of inertia for rotations of the body around the centre of the sphere and $I_{ax}$ is the moment of inertia for rotations around the symmetry axis of the body.

    We find the explicit form of $\dif L(p, (\vec{n}, \vec{v}), r) = (p, (\vec{n}, \vec{v}, \vec{\alpha}, \vec{\beta}), r, a, b) \in \imath_{\T^*\T\fr}(\T^*\T\fr)$ evaluating it on a vector $(p, (\vec{n}, \vec{v}, \vec{u}, \vec{w}), r, k, s) \in \imath_{\T\T\fr}(\T\T\fr)$ represented by the curve $(\exp(tY)p, (\exp(tY)\vec{n}, (1 + yt)\exp(tY)\exp(txN)\vec{v}), r + ts) \in \imath_{\T\fr}(\T \fr)$, where $Yp \simeq (p, Y\vec{n}, k)$, $\vec{u} = Y\vec{n}$, $\vec{w} = (Y + xN + y)\vec{v}$:
    \begin{align*}
        \dif L (p, (\vec{n}, \vec{v}, \vec{u}, \vec{w}), r, k, s) &\equiv \der\, L(\exp(tY)p, (\exp(tY)\vec{n}, (1 + yt)\exp(tY)\exp(txN)\vec{v}), r + ts) \\
        &= I_\perp \bra{\vec{v}}\ket{y\vec{v}} + I_{ax}rs .
    \end{align*}
    This leads us to $\dif L(p, (\vec{n}, \vec{v}), r) = (p, (\vec{n}, \vec{v}, c\vec{n}, I_\perp  \vec{v}), r, 0, I_{ax}r)$, where $c \in \R$ and $c\vec{n}$ is a representative of the annihilator of $\{Y\vec{n} \}_{Y \in so(3)}$. Therefore, the dynamics takes the form
    \begin{align*}
        \widetilde{\mathcal{D}} = \left\{ \left(p, (\vec{n}, I_\perp \vec{v}, \vec{v}, c\vec{n} - I_{ax}r\vec{v}\times\vec{n}), I_{ax}r, r, 0 \right) \right\},
    \end{align*}
    and the Euler-Lagrange equations for this system are
    \begin{align}
        \label{eq:E-L_frame}
        I_\perp \der\, \vec{v} &= c\vec{n} - I_{ax}r\vec{v}\times\vec{n}, \\
        \nonumber I_{ax}\der\, r &= 0 .
    \end{align}
    In the first equation, as expected, $c\vec{n}$ corresponds to the centrifugal force, and $I_{ax}r\vec{v}\times\vec{n}$ describes the gyroscopic precession enforced by the symmetry axis revolving around the sphere.

\section{Atiyah algebroid with connection and reduced Tulczyjew triple}

In Section \ref{sec:algebroids_and_triples}, we recalled how the framework of the Tulczyjew triple allows us to encode a Lie algebroid structure in one double vector bundle morphism.

In this section, we describe the reduction of the Tulczyjew triple on a principal bundle by the action of the tangent group $\T G$. From such a reduction emerges a particular Lie algebroid -- the Atiyah algebroid. We will show that in the presence of a connection on the principal bundle, the Atiyah algebroid is also trivialisable, and we will present the trivialised double vector bundle morphism encoding it.

\vspace{\spacesmall}

The \textit{Atiyah algebroid} $(A(P), M, [\cdot, \cdot]_{A(P)}, \rho_{A(P)})$ of a principal bundle $P$ is a Lie algebroid built on the vector bundle $\pi_{A(P)} : A(P) := \T P / G \longrightarrow M$ obtained by dividing $P$ by the tangent action of $G$. Its sections may be identified with the $G$-invariant vector fields on $P$, and the Lie bracket $[\cdot, \cdot]_{A(P)}$ is, therefore, induced from the vector field bracket on $P$. The anchor $\rho_{A(P)}$ is induced by $\T \pi_P$. We denote the projection $\T P \rightarrow A(P)$ with $\delta_{A(P)}$.

It is a well-known fact that a principal connection induces a splitting of the anchor $\rho_{A(P)}$ leading to a trivialisation of the Lie algebroid $\imath_{A(P)}: A(P) \simeq \T M \oplus_M \ad P$ (see \cite{atiyah}). However, this direct sum is in the category of Lie algebroids over $M$ only if the connection is flat -- otherwise, it is just a sum of vector bundles. We note that the above decomposition agrees with the trivialisation \eqref{eq:triv_TP} divided by $G$. We denote the trivialised projection onto the quotient with $\quo$.

Starting with the cotangent of a principal bundle and dividing it by the dual of the tangent action of $G$, we get the \textit{dual of the Atiyah algebroid}, which, if there is a connection, in a similar manner may be decomposed as $\imath_{A^*(P)} : A^*(P) \simeq \ad^*P \oplus_M \T^* M$, where $\ad^* P$ denotes the dual of the adjoint bundle (the defining action on $\alg^*$ is $\AdG^*$). The projections are denoted with $\delta_{A^*(P)}$ and $\quoD$.

\subsection{Reduction of trivialised \textsf{TT}\textsl{P}}
    To reduce the canonical flip, we first find the explicit form of a trivialisation $\T A(P) \simeq \T(\T M \oplus_M \adP) \simeq \T\T M \times_M \adP \times_M \adP$ and of the projection from $\imath_{\T\T P}(\T\T P)$ onto it.

    \label{con:TAP_triv}
    By definition, $A(P)$ is the space of orbits of the tangent action of $G$ on $\T P$. It, along with the version with the functor $\T$ applied, may be diagrammed as
    \begin{equation}
        \label{xy:TG_action}
        \xymatrix@C-15pt@R-10pt{
             \T P \times G \ar[rr]^{\T \wp} \ar[d]_{\delta_{A(P)} \circ \pr_1} & & \T P \ar[d]^{\delta_{A(P)}} \\
             A(P) \ar@{=}[rr] & & A(P) \;,
         }
         \quad
         \xymatrix@C-15pt@R-10pt{
             \T\T P \times \T G \ar[rr]^{\T(\T \wp)} \ar[d]_{\T\delta_{A(P)} \circ \pr_1} & & \T\T P \ar[d]^{\T\delta_{A(P)}} \\
             \T A(P) \ar@{=}[rr] & & \T A(P) \; ,
         }
    \end{equation}
    where we understand the tangent action as a map: $\T P \times G \longrightarrow \T P$. The latter diagram expresses the fact that the tangent of the Atiyah algebroid is the space of orbits of the action of $\T G$ on $\T\T P$. The projection onto the quotient takes the form of $\T\delta_{A(P)}$. In the following, we examine a trivialisation of this diagram, which leads to a convenient trivialisation of $\T A(P)$.

    \begin{remark}
        The group structure of the tangent group $\T G$ carries to the trivialisation $G \times \alg$ as a semidirect product.
    \end{remark}

    We start by considering a curve $\R \ni t \longmapsto \gamma_g(t) = g\exp(tW) \in G$ representing a generic vector $(g, W) \in \imath_{\T G}(\T G) = G \times \alg$ and a curve $\R \ni t \longmapsto (\gamma_p(t), \gamma_v(t), \gamma_X(t)) \in \imath_{\T P}(\T P)$ representing a vector $(p, \V, X, Y, Z) \in \imath_{\T\T P}(\T\T P)$, i.e.
    \begin{align*}
        \T P \ni \dot{\gamma}_p(0) &\simeq (p, \T\tau_M(\V), Y) \in P \times_M \T M \times \alg, \\
        \T\T M \ni \dot{\gamma}_v(0) &= \V, \\
        \T \alg \ni \dot{\gamma}_X(0) &= \der \, X + tZ \simeq (X, Z) \in \alg \times \alg.
    \end{align*}
    We now calculate the right action of $(g, W)$ on $(p, \V, X, Y, Z)$, that is
    \begin{align*}
        (p, \V, X, Y, Z) \cdot (g, W) &= \der \, (\gamma_p(t) \cdot \gamma_g(t), \gamma_v(t), \AdG_{\gamma_g(t)^{-1}}(\gamma_X(t))).
    \end{align*}

    \begin{itemize}
        \item The tangent to the $P$-curve:
            \begin{multline*}
                \der \, \gamma_p(t) \cdot \gamma_g(t) = \T_p\wp_g(\dot{\gamma}_p(0)) + \sigma(W)_{pg} \simeq \\ (pg, \T\tau_M(\V), \AdG_{g^{-1}}Y + W) \in P \times_M \T M \times \alg,
            \end{multline*}
        \item and the tangent to the $\alg$-curve:
            \begin{align*}
                \der \, \AdG_{\gamma_g(t)^{-1}}(\gamma_X(t)) &\simeq  (\AdG_{g^{-1}}X, [\AdG_{g^{-1}}X, W]_\alg + \AdG_{g^{-1}}Z) \in \alg \times \alg.
            \end{align*}
    \end{itemize}
    Composing the parts back together (additionally, we omitted the redundant component, and changed the order to fit the convention), we get the ($G \times \alg$)-equivalence between the elements (via respective elements of $G \times \alg$: $(g, W)$, $(e, -Y)$)
    \begin{align}
        \label{eq:TAP_equiv}
        \notag(p, \V, X, Y, Z) &\sim_{(G \times \alg)} \big(pg; \V; \AdG_{g^{-1}}X; \AdG_{g^{-1}}Y + W; [\AdG_{g^{-1}}X, W]_\alg + \AdG_{g^{-1}}Z\big) \\
        &\sim_{(G \times \alg)} \big(p; \V; X; 0; [Y, X]_\alg + Z\big).
    \end{align}

    We now show the foretold one-to-one correspondence between orbits of the $(G \times \alg)$-action on $\imath_{\T\T P}(\T\T P)$ and elements of $\T\T M \times_M \ad P \times_M \ad P$. First, inspired by the form of the last representative in (\ref{eq:TAP_equiv}), we postulate the following projection onto trivialised orbits
    \begin{align*}
        \label{eq:triv_TAP}
        \widetilde{\T\delta}_{A(P)} &: \imath_{\T\T P}(\T\T P) \longrightarrow \T\T M \times_M \ad P \times_M \ad P \\
        &: (p, \V, X, Y, Z) \longmapsto (\V, [(p, X)]_{\ad P}, [(p, [Y, X]_\alg + Z)]_{\ad P}) .
    \end{align*}
    One can easily check that the projection is $(G \times \alg)$-invariant. Therefore, we have the well-defined map
    \begin{align*}
        \mu_{\T A(P)} &: \imath_{\T\T P}(\T\T P)/(G \times \alg) \longrightarrow \T\T M \times_M \ad P \times_M \ad P \\
        &: [(p, \V, X, Y, Z)]_{(G \times \alg)} \longmapsto (\V, [(p, X)]_{\ad P}, [(p, [Y, X]_\alg + Z)]_{\ad P}) .
    \end{align*}
    What is more, it is easy to check that the following is its well-defined inverse
    \begin{align*}
        \jmath_{\T A(P)} &: \T\T M \times_M \ad P \times_M \ad P \longrightarrow \imath_{\T\T P}(\T\T P)/(G \times \alg) \\
        &: (\V, [p, X]_{\ad P}, [p, T]_{\ad P}) \longmapsto [(p, \V, X, 0, T)]_{(G \times \alg)}.
    \end{align*}

    Finally, using the above results and the commutativity of \eqref{xy:TG_action}, we conclude that $\T A(P)$ is indeed trivialisable to $\T\T M \times_M \ad P \times_M \ad P$. We denote this trivialisation by
    \begin{equation}
        \label{eq:triv_TAP}
        \imath_{\T A(P)}: \T A(P) \simeq \T\T M \times_M \ad P \times_M \ad P.
    \end{equation}

    \begin{remark}
        \label{cor:triv_TAP}
        The morphism $\imath_{\T A(P)}$ is a double vector bundle isomorphism. Both vector structures of $\T\T M \times_M \ad P \times_M \ad P$ are inherited from $\imath_{\T\T P}(\T\T P)$ via $\widetilde{\T\delta}_{A(P)}$ and read
        \begin{equation*}\begin{gathered}
            \scalebox{1}{\xymatrix@C-15pt@R-5pt{
                 & \T\T M \times_M \ad P \times_M \ad P \ar[dl]|{(\tau_{\T M} \circ \pr_1, \pr_2)} \ar[dr]|{\T \tau_M \circ \pr_1} & \\
                \T M \times_M \ad P \ar[dr]|{\tau_M \circ \pr_1} & \T M \times_M \ad P \ar[d]^(0.3){\tau_M \circ \pr_1} \ar@{^{(}->}[u] & \T M \ar[dl]|{\tau_M} \\
                 & M & .
            }}
        \end{gathered}\end{equation*}
        This will be the case for all the incoming reductions of iterated tangent and cotangent bundles.
    \end{remark}

\subsection{Reduction of the trivialised canonical flip}
    Having the trivialisation of $\T A(P)$ computed, the idea behind obtaining the reduced canonical flip in trivialisation is fairly simple and diagrammed as
    \begin{equation*}\begin{gathered}
        \scalebox{1}{\xymatrix{
             P \times_M \T\T M \times \alg \times \textcolor{blue}{\alg} \times \textcolor{red}{\alg} \ar[rr]^{\tilde{\kappa}_P} \ar[d]|{\widetilde{\T\delta}_{A(P)}} & & P \times_M \T\T M \times \alg \times \textcolor{blue}{\alg} \times \textcolor{red}{\alg} \ar[d]|{\widetilde{\T\delta}_{A(P)}} \\
             \T\T M \times_M \ad P \times_M \ad P \ar@{-|>}[rr]_{\tilde{\kappa}_{A(P)}} & & \T\T M \times_M \ad P \times_M \ad P \;,
         }}
    \end{gathered}\end{equation*}
    where we only need to be careful whether the reduced canonical flip is a map or just a relation (which we mark with a triangular arrow). Inserting a vector into the top-left entry of the diagram, we calculate that the canonical flip on the Atiyah algebroid (reduced canonical flip on a principal bundle) is the relation
    \begin{gather*}
        \T\T M \times_M \ad P \times_M \ad P \ni (\V, \mathcal{X}, \mathcal{Z}_1) \sim_{\tilde{\kappa}_{A(P)}} (\mathcal{U}, \mathcal{Y}, \mathcal{Z}_2) \in \T\T M \times_M \ad P \times_M \ad P \\
        \iff \\
        \mathcal{U} = \kappa_M(\V) \quad \text{and} \quad \mathcal{Z}_2 = \mathcal{Z}_1 + [\mathcal{X}, \mathcal{Y}]_\adP + [(p, \Omega_p\big(\tau_{\T M}(\V), \T\tau_{M}(\V)\big)]_\adP.
    \end{gather*}

    \begin{remark}
        We note that, upon choosing the base points with respect to the second vector structure $\T\tau_M \circ \pr_1$ on the left and the first vector structure $(\tau_{\T M} \circ \pr_1, \pr_2)$ on the right, $\tilde{\kappa}_{A(P)}$ becomes a morphism of vector spaces.
    \end{remark} 

\subsection{Reduction of trivialised \textsf{TT}*\textsl{P}}
    In order to determine the reduced Tulczyjew isomorphism, we need to first reduce the dual bundles $\T\T^* P$ and $\T^* \T P$.

    The construction of $\imath_{\T A^*(P)}$ is analogical to Section \ref{con:TAP_triv}. $A^*(P)$ is the quotient space by the (dual to the tangent) action of $G$. Therefore, applying the tangent functor, $\T A^*(P)$ is the quotient of $\T\T^* P$ by the action of $\T G$. We need to obtain the explicit form of this action. This time, we start with a curve $\R \ni t \longmapsto (\gamma_p(t), \gamma_\alpha(t), \gamma_A(t)) \in \imath_{\T^* P}(\T^* P)$ representing a vector $(p, \dot{\alpha}, A, Y, B) \in P \times_M \T\T^* M \times \alg^* \times \alg \times \alg^*$ and, once again, a curve $\R \ni t \longmapsto \gamma_g(t) = g\exp(tW) \in G$ representing a generic vector $(g, W) \in G \times \alg$. Therefore,
    \begin{align*}
        (p, \dot{\alpha}, A, Y, B) \cdot (g, W) = \der \, \big(\gamma_p(t) \cdot \gamma_g(t), \gamma_\alpha(t), \AdG_{\gamma_g(t)}^*(\gamma_A(t))\big).
    \end{align*}
    Differentiating (and using the identity $\der \, \AdG^*_{g\exp(tW)}(A + tB) = \AdG^*_g(B) + (\AdG_g \circ \ad_W)^*(A)$), we get different elements in the same orbit of the $(G \times \alg)$-action:
    \begin{align*}
        (p, \dot{\alpha}, A, Y, B) &\sim_{(G \times \alg)} (pg, \dot{\alpha}, \AdG_g^*(A), \AdG_{g^{-1}}(Y) + W, \AdG_g^*(B) + (\AdG_g \circ \ad_W)^*(A)) \\
        &\sim_{(G \times \alg)} (p, \dot{\alpha}, A, 0, B - \ad_Y^*(A)).
    \end{align*}
    Next, we define the $(G \times \alg)$-invariant projection
    \begin{align*}
        \widetilde{\T\delta}_{A^*(P)} &: \imath_{\T\T^* P}(\T\T^* P) \longrightarrow \T\T^* M \times_M \ad^* P \times_M \ad^* P \\
        &: (p, \dot{\alpha}, A, Y, B) \longmapsto (\dot{\alpha}, [(p, A)]_{\ad^* P}, [(p, B - \ad_Y^*(A))]_{\ad^* P}) .
    \end{align*}
    Hence, we have the map
    \begin{align*}
        \mu_{\T A^*(P)} &: \imath_{\T\T^* P}(\T\T^* P)/(G \times \alg) \longrightarrow \T\T^* M \times_M \ad^* P \times_M \ad^* P \\
        &: [(p, \dot{\alpha}, A, Y, B)]_{(G \times \alg)} \longmapsto (\dot{\alpha}, [(p, A)]_{\ad^* P}, [(p, B - \ad_Y^*(A))]_{\ad^* P})
    \end{align*}
    with the inverse
    \begin{align*}
        \jmath_{\T A^*(P)} &: \T\T^* M \times_M \ad^* P \times_M \ad^* P \longrightarrow \imath_{\T\T^* P}(\T\T^* P)/(G \times \alg) \\
        &: (\dot{\alpha}, [(p, A)]_{\ad^* P}, [(p, B)]_{\ad^* P}) \longmapsto [(p, \dot{\alpha}, A, 0, B)]_{(G \times \alg)}.
    \end{align*}
    Therefore, we have obtained a one-to-one correspondence between orbits of the $(G \times \alg)$-action on $\imath_{\T\T^* P}(\T\T^* P)$ and elements of $\T\T^* M \times_M \ad^* P \times_M \ad^* P$, which in turn establishes a trivialisation $\imath_{\T A^*(P)} : \T A^*(P) \simeq \T\T^* M \times_M \ad^* P \times_M \ad^* P$.

\subsection{Reduction of trivialised \textsf{T}*\textsf{T}\textsl{P}}
    This time, we look for a trivialisation in the form $\T^* A(P) \simeq \T^*(\T M \oplus_M \adP) \simeq \T^*\T M \times_M \adP \times_M \ad^*P$.

    First, we take elements of $\T^* \T P$ and notice that, at any $v \in \T P$, it only makes sense to consider covectors having a constant evaluation on all elements of $\T_v\T P$ projecting to the same one in $\T A(P)$, i.e. orbits of the action of $\{e\} \times \alg \subset G \times \alg$. In other words, reducible covectors form the annihilator $(\ker \T \delta_{A(P)})^\circ$ at every point of $\T P$. Therefore, we should find the explicit form of the annihilator in trivialisation, starting with the kernel of $\widetilde{\T\delta}_{A(P)}$. From $\widetilde{\T\delta}_{A(P)}(p, \V, X, Y, Z) = (\V, \mathcal{X}, [(p, [Y, X]_\alg + Z)]_{\ad P})$, it follows that
    \begin{align*}
        \Ker \widetilde{\T\delta}_{A(P)} = \left\{ (p, \V, X, Y, Z) \in \imath_{\T\T P}(\T\T P) \; \mid \; \V = 0 \; \text{and} \; Z = -[Y, X]_\alg \right\}.
    \end{align*}
    Hence, the reducible covectors over $(p, \pi_{\T M}(\omega), X)$ satisfy the condition
    \begin{gather*}
        \langle \omega, 0 \rangle + \langle B, Y \rangle + \langle C, -[Y, X]_\alg \rangle = 0 .
    \end{gather*}
    Therefore, the coisotropic submanifold of reducible trivialised covectors collected over the base takes the form
    \begin{align*}
        \tilde{K} = \big\{ (p, \omega, X, B, C) \in P \times_M \T^*\T M \times \alg \times \alg^* \times \alg^* \; \mid \; B = -\ad_X^*(C) \big\}.
    \end{align*}

    Second, we need to ensure that trivialised elements of $A^*(P)$ are constant also on orbits of the subgroup $G \times \{ 0 \} \subset G \times \alg$. It yields a relation on $\tilde{K}$, which equivalence classes we find using an evaluation on two vectors from the same orbit:
    \begin{gather*}
        \begin{aligned}
            \langle (\omega, B, C), (\V, Y, Z) \rangle_{} &= \langle (\omega', B', C'), (\V, \AdG_{g^{-1}}(Y), \AdG_{g^{-1}}(Z)) \rangle,
        \end{aligned}
    \end{gather*}
    where the left evaluation is at the point $(p, v, X)$ and the right at $(pg, v, \AdG_{g^{-1}}(X))$. Comparing both sides, we get
    \begin{align*}
        (p, \omega, X, -\ad_X^*(C), C) \sim_G \big(pg, \omega, \AdG_{g^{-1}}(X), -(\ad_X \circ \AdG_g)^*(C), \AdG_g^*(C)\big).
    \end{align*}
    We denote the submanifold $\tilde{K}$ divided by the relation by $\tilde{K}/G$ .

    Finally, we propose the projection map (inspired by the fact that the fourth component in the equivalence class can be reassembled from the third and the fifth)
    \begin{align*}
        \widetilde{\T^*\delta}_{A(P)} &: \tilde{K} \longrightarrow \T^*\T M \times_M \ad P \times_M \ad^* P \\
        &: (p, \omega, X, -\ad_X^*(C), C) \longmapsto (\omega, [(p, X)]_{\ad P}, [(p, C)]_{\ad^* P}),
    \end{align*}
    which $G$-invariance is proven by
    \begin{equation*}\begin{gathered}
        \scalebox{0.93}{\xymatrix{
            (p, \omega, X, -\ad_X^*(C), C) \ar@{|->}[r]^(0.4){\cdot g} \ar@{|->}[d]^{\widetilde{\T^*\delta}_{A(P)}} & \big(pg, \omega, \AdG_{g^{-1}}(X), -(\ad_X \circ \AdG_g)^*(C), \AdG_g^*(C)\big) \ar@{|->}[d]^{\widetilde{\T^*\delta}_{A(P)}} \\
            (\omega, [(p, X)]_{\ad P}, [(p, C)]_{\ad^* P}) \ar@{=}[r] & (\omega, [(pg, \AdG_{g^{-1}}(X))]_{\ad P}, [(pg, \AdG_g^*(C))]_{\ad^* P}),
        }}
    \end{gathered}\end{equation*}
    and which induces the well-defined trivialisation
    \begin{align*}
        \mu_{\T^*A(P)} &: \tilde{K}/G \longrightarrow \T^*\T M \times_M \ad P \times_M \ad^* P \\
        &: [(p, \omega, X, -\ad_X^*(C), C)]_G \longmapsto (\omega, [(p, X)]_{\ad P}, [(p, C)]_{\ad^* P})
    \end{align*}
    with the apparent inverse
    \begin{align*}
        \jmath_{\T^* A(P)} &: \T^*\T M \times_M \ad P \times_M \ad^* P \longrightarrow \tilde{K}/G \\
        &: (\omega, [(p, X)]_{\ad P}, [(p, C)]_{\ad^* P}) \longmapsto [(p, \omega, X, -\ad_X^*(C), C)]_G.
    \end{align*}
    As introduced at the beginning, the described construction is directly the trivialised counterpart of the construction of $\T^* A(P)$, and we conclude that there exists a trivialisation $\imath_{\T^* A(P)} : \T^* A(P) \simeq \T^*\T M \times_M \ad P \times_M \ad^* P$.

    \begin{remark}
        Upon taking the appropriately regular value of the momentum map corresponding to the $G$-action, the reduction coincides with the symplectic reduction by the action of $G$. Therefore, there is a unique induced symplectic form on $\T^* A(P)$.
    \end{remark}

\subsection{Reduction of the trivialised Tulczyjew isomorphism}
    Having the trivialisations computed, we can find the form of the reduced Tulczyjew isomorphism in a similar manner to the reduced canonical flip. This time, we use the diagram
    \begin{equation*}\begin{gathered}
        \scalebox{1}{\xymatrix{
            \tilde{K} \ar[rr]^{\tilde{\alpha}^{-1}_P} \save[]+<0cm,1.8pc>*\txt<10pc>{$P \times_M \T^*\T M \times \alg \times {\alg}^* \times {\alg}^*$ \\ $\cup$} \restore \ar[d]|{\widetilde{\T^*\delta}_{A(P)}} & & P \times_M \T\T^* M \times \alg^* \times \alg \times {\alg}^* \ar[d]|{\widetilde{\T\delta}_{A^*(P)}} \\
            \T^*\T M \times_M \ad P \times_M \ad^* P & & \T\T^* M \times_M \ad^* P \times_M \ad^* P \ar@{-|>}[ll]_{\tilde{\alpha}_{A(P)}} \; .
         }}
    \end{gathered}\end{equation*}
    $\tilde{\alpha}_A(P)$ is only a relation, but the reduction of $\tilde{\alpha}_P^{-1}$ happens to be a mapping. Therefore, in the context of the Atiyah algebroid, we use the term {\sl Tulczyjew map} for the reduced inverse of the trivialised Tulczyjew isomorphism on a principal bundle:
    \begin{align*}
        \tilde{\varepsilon}_{A(P)} &: \T^*\T M \times_M \adP \times_M \ad^*P \longrightarrow \T\T^* M \times_M \ad^*P \times_M \ad^*P \\
        &: (\omega, [(p, X)]_{\ad P}, [(p, C)]_{\ad^* P}) \\
        &\quad \longmapsto \left(\alpha_M^{-1}(\omega) \dot{-} \Omega^*_{(p, \omega)}(C), [(p, C)]_{\ad^* P}, [(p, -\ad_X^*(C))]_{\ad^* P}\right).
    \end{align*}

    \begin{remark}
        The above mapping is well-defined because $$\Omega^*_{(p, \omega)}(C) \equiv \Omega\big(\pi_{\T M}(\omega)^\Hor_p, \T\tau_M(\,\cdot\,)^\Hor_p \big)^*(C) = \Omega\big(\pi_{\T M}(\omega)^\Hor_{pg}, \T\tau_M(\,\cdot\,)^\Hor_{pg} \big)^*(\AdG^*_g(C)).$$
    \end{remark}

\subsection{Trivialised Atiyah algebroid via the Tulczyjew mapping}
    As recalled in the introduction, every Lie algebroid structure on a vector bundle $E$ is encoded in a particular double vector bundle morphism $\varepsilon : \T^* E \rightarrow \T E^*$. The structure of the trivialised Atiyah algebroid on $\T M \times_M \ad P$ is encoded in the Tulczyjew map $\widetilde{\varepsilon}_{A(P)}$:
    \begin{equation*}\begin{gathered}
        \scalebox{1}{\xymatrix@C-15pt@R-10pt{
             & \T^*\T M \times_M \ad P \times_M \ad^* P \ar[rr]^{\tilde{\varepsilon}_{A(P)}} \ar[dl]|{(\pi_{\T M} \circ \pr_1, \pr_2)} \ar[dd]|(0.3){(\xi_{\T^* M} \circ \pr_1, \pr_3)} & & \T\T^* M \times_M \ad^* P \times_M \ad^* P  \ar[dl]|{\T \pi_M \circ \pr_1} \ar[dd]|{(\tau_{\T^* M} \circ \pr_1, \pr_2)} & \\
            \T M \times_M \ad P \ar[rr]_(0.75){\tilde{\rho}_{A(P)}} \ar[ddr] & & \T M \ar[ddr] & & \\
             & \T^* M \times_M \ad^* P \ar@{=}[rr] \ar[d] & & \T^* M \times_M \ad^* P \ar[d] \\
             & M \ar@{=}[rr] & & M & ,
        }}
    \end{gathered}\end{equation*}
    where $\tilde{\rho}_{A(P)} := \pr_1$ is the trivialised anchor of the Lie algebroid. Additionally, as the canonical flip relation is the dual of the Tulczyjew map, it also encodes the algebroid structure.

\subsection{Reduction of trivialised \textsf{T}*\textsf{T}*\textsl{P}}
    We are left with the reduction of $\imath_{\T^*\T^* P}(\T^*\T^* P)$. This reduction allows us to also calculate the last part of the Tulczyjew triple on the Atiyah algebroid.

    The construction of the trivialisation map $\imath_{\T^* A^*(P)}$ is fully analogous to the construction of $\imath_{\T^* A(P)}$, and, as such, we only highlight its most important steps. First, the coisotropic submanifold of the reducible covectors is of the form
    \begin{align*}
        \tilde{L} = \big\{ ((p, \psi, A, B, Z)) \in P \times_M \T^*\T^* M \times \alg^* \times \alg^* \times \alg \; \mid \; B = \ad_Z^*(A) \big\}.
    \end{align*}
    Second, due to the equivalence
    \begin{align*}
        (p, \psi, A, \ad_Z^*(A), Z) \sim_G (pg, \psi, \AdG_g^*(A), (\ad_Z \circ \AdG_g)^*(A), \AdG_{g^{-1}}(Z)),
    \end{align*}
    we introduce the $G$-invariant projection
    \begin{align*}
        \widetilde{\T^*\delta}_{A^*(P)} &: \tilde{L} \longrightarrow \T^*\T^* M \times_M \ad^* P \times_M \ad P \\
        &: (p, \psi, A, \ad_Z^*(A), Z) \longmapsto (\psi, [(p, A)]_{\ad^* P}, [(p, Z)]_{\ad P})
    \end{align*}
    and the obviously well-defined and invertible
    \begin{align*}
        \mu_{\T^*A^*(P)} &: \tilde{L}/G \longrightarrow \T^*\T^* M \times_M \ad^* P \times_M \ad P \\
        &: [(p, \psi, A, \ad_Z^*(A), Z)]_G \longmapsto (\psi, [(p, A)]_{\ad^* P}, [(p, Z)]_{\ad P}).
    \end{align*}
    The above yields a trivialisation $\imath_{\T^* A^*(P)} : \T^* A^*(P) \simeq \tilde{L}/G \simeq \T^*\T^* M \times_M \ad^* P \times_M \ad P$.

\subsection{Reduction of the trivialised symplectic-based isomorphism}
    Once again, we obtain the reduced symplectic-based isomorphism of the Tulczyjew triple using the diagrammatic presentation
    \begin{equation*}\begin{gathered}
        \scalebox{1}{\xymatrix{
            P \times_M \T\T^* M \times \alg^* \times \alg \times {\alg}^* \ar[d]|{\widetilde{\T\delta}_{A^*(P)}} & & \tilde{L} \ar[ll]_{\tilde{\beta}^{-1}_P} \save[]+<0cm,1.8pc>*\txt<10pc>{$P \times_M \T^*\T^* M \times \alg^* \times \alg^* \times \alg$ \\ $\cup$} \restore \ar[d]|{\widetilde{\T^*\delta}_{A^*(P)}} \\
            \T\T^* M \times_M \ad^* P \times_M \ad^* P \ar@{-|>}[rr]_{\tilde{\beta}_{A(P)}} & & \T^*\T^* M \times_M \ad^* P \times_M \ad P \; .
         }}
    \end{gathered}\end{equation*}

    This time also, the reduced version of $\tilde{\beta}_P^{-1}$, which we denote by $\tilde{\eta}_{A(P)}$, is a mapping, but the reduced version of the symplectic-based isomorphism is a relation. The mapping takes the form
    \begin{align*}
        \tilde{\eta}_{A(P)} &: \T^*\T^* M \times_M \ad^*P \times_M \adP \longrightarrow \T\T^* M \times_M \ad^*P \times_M \ad^*P \\
        &: (\Theta, [(p, A)]_{\ad^* P}, [(p, Z)]_{\ad P}) \\
        &\quad \longmapsto \left(\beta_M^{-1}(\Theta) \dot{-} \Omega^*_{(p, \Theta)}(A), [(p, A)]_{\ad^* P}, [(p, -\ad_Z^*(A))]_{\ad^* P}\right).
    \end{align*} 

\subsection{Dynamics on the trivialised Atiyah algebroid}
    We summarise the reduction of the Tulczyjew triple on a principal bundle with connection using the diagram of the reduced triple highlighting some of the trivialised double vector structure:
    \begin{equation}\begin{gathered}
        \label{xy:triple_on_A(P)}
        \scalebox{1}{\xymatrix@C-10pt@R-5pt{
            \T^*\T^* M \times_M \ad^* P \times_M \ad P \ar[dr]|{(\pi_{\T^* M} \circ \pr_1, \pr_2)} & & \\
             & \T^* M \times_M \ad^* P \ar[dr]|{\pi_M \circ \pr_1} & \\
            \T\T^* M \times_M \ad^* P \times_M \ad^* P \ar@{-|>}[uu]_{\tilde{\beta}_{A(P)}} \ar@{-|>}[dd]^{\tilde{\alpha}_{A(P)}} \ar[ur]|{(\tau_{\T^* M} \circ \pr_1, \pr_2)} \ar[dr]|{\T\pi_M \circ \pr_1} & & M \\
             & \T M \ar[ur]|{\tau_M} & \\
            \T^*\T M \times_M \ad P \times_M \ad^* P \ar@{-->}[ur]|{\widetilde{\rho}_{A(P)} \circ (\pi_{\T M} \circ \pr_1, \pr_2)} & & .
         }}
    \end{gathered}\end{equation}
    The dashed arrow indicates that the projection is not a canonical projection corresponding to a vector structure, but, in fact, one of the vector projections composed with the anchor of the Atiyah algebroid.

    Now, we consider a mechanical system modelled by the Atiyah algebroid with a trivialised Lagrangian invariant under the tangent action, i.e. $\widetilde{L}(pg, v, \AdG_{g^{-1}}(X)) = \widetilde{L}(p, v, X)$. In such a case, the reduced Lagrangian $\tilde{l} : \T M \oplus_M \adP \rightarrow \R$ may be introduced as $\tilde{l}(v, [(p, X)]_{\adP}) := \widetilde{L}(p, v, X)$. The reduced dynamics $\widetilde{d} := \widetilde{\varepsilon}_{A(P)} \circ \dif \tilde{l} \in \T\T^* M \times_M \ad^*P \times \ad^* P$ may be written as
    \begin{gather*}
        \dif \tilde{l} (v, [(p, X)]_{\adP}) = (\partial_{\T M} \widetilde{L}_V, [(p, X)]_\adP, [(p, \partial_\alg \widetilde{L}_V)]_{\ad^* P}) \in \T^*\T M \times_M \adP \times_M \ad^*P, \\
        \widetilde{d} = \left\{ \left. \begin{gathered}(\varphi, \mathcal{A}, \mathcal{B}) \in \\
        \T\T^* M \times_M \ad^*P \times_M \ad^*P \end{gathered} \right| \begin{aligned} \varphi &= \alpha_M^{-1}(\partial_{\T M} \widetilde{L}_V) \dot{-} \Omega^*_{(p, \partial_{\T M} \widetilde{L}_V)}(\partial_\alg \widetilde{L}_V), \\
        \mathcal{A} &= [(p, \partial_\alg \widetilde{L}_V)]_{\ad^* P}), \; \mathcal{B} = [(p, -\ad_X^*(\partial_\alg \widetilde{L}_V)]_{\ad^* P})) \\ &\text{for some } V \equiv (p, v, X) \in P\times_M \T M \times \alg
        \end{aligned} \right\},
    \end{gather*}
    where $\partial_{\T M} \widetilde{L}_V$ and $\partial_\alg \widetilde{L}_V$ are the same as in \eqref{eq:dynamics}.

    On the other hand, if we have a Hamiltonian invariant under the dual of the tangent action, we can introduce its reduced version $\widetilde{h}(\alpha, [(p, A)]_{\ad^*P}) := \widetilde{H}(p, \alpha, A) \equiv \widetilde{H}(\mathfrak{A})$. The Hamiltonian vector field takes the form
    \begin{align*}
        \widetilde{X}_h(\alpha, [(p, A)]_{\ad^*P}) := \tilde{\eta}_{A(P)} \circ \dif \widetilde{h} (\alpha, [(p, A)]_{\ad^*P}) = \tilde{\eta}_{A(P)} (\partial_{\T^* M} \widetilde{H}_\mathfrak{A}, [(p, A)]_{\ad^*P}, [(p, \partial_{\alg^*} \widetilde{H}_\mathfrak{A})]_\adP),
    \end{align*}
    where $\partial_{\T^* M} \widetilde{H}_\mathfrak{A}$ and $\partial_{\alg^*} \widetilde{H}_\mathfrak{A}$ are, once again, taken from \eqref{eq:dynamics_ham}. Therefore, the dynamics obtained from the Hamiltonian side is as follows:
    \begin{gather*}
        \widetilde{d} = \left\{ \left. \begin{gathered}(\varphi, \mathcal{A}, \mathcal{B}) \in \\
        \T\T^* M \times_M \ad^*P \times_M \ad^*P \end{gathered} \right| \begin{aligned} \varphi &= \beta_M^{-1}(\partial_{\T^* M} \widetilde{H}_\mathfrak{A}) \dot{-} \Omega^*_{(p, \partial_{\T^* M} \widetilde{H}_\mathfrak{A})}(A), \\
        \mathcal{A} &= [(p, A)]_{\ad^* P}, \quad \mathcal{B} = [(p, -\ad_{\partial_{\alg^*} \widetilde{H}_\mathfrak{A}}^*(A)]_{\ad^* P} \\ &\text{for some } \mathfrak{A} \equiv (p, \alpha, A) \in P\times_M \T^* M \times \alg^*
        \end{aligned} \right\}.
    \end{gather*} 

\subsection{Example: Reduction of the triple on the frame bundle of a sphere}
    In this section, we showcase the reduction algorithm on the example of the Tulczyjew triple on the frame bundle $\fr$ calculated in Section \ref{sec:example1}.

    First, by \eqref{eq:triv_TAP}, we state a trivialisation of the Atiyah algebroid of the frame bundle
    \begin{align*}
        \T A(\fr) \simeq \T\T S^2 \times_{S^2} \ad \fr \times_{S^2} \ad \fr.
    \end{align*}
    Next, we notice that, as the algebra $so(2)$ is commutative, the adjoint bundle is trivialisable -- $\ad \fr \simeq S^2 \times so(2)$. Hence, we may further simplify $\T A(\fr)$ as
    \begin{align*}
        \T A(\fr) \simeq \T\T S^2 \times_{S^2} (S^2 \times so(2)) \times_{S^2} (S^2 \times so(2)) \simeq \T\T S^2 \times so(2) \times so(2).
    \end{align*}
    The projection from $\imath_{\T\T \fr}(\T\T\fr)$ onto trivialised orbits of the $(SO(2) \times so(2))$-action takes the simple form
    \begin{align*}
        \mu_{\T A(\fr)}((\vec{n}, f), \V, x, y, z) = \left( \V, x, z \right) \in \T\T S^2 \times so(2) \times so(2).
    \end{align*}

    Similarly, using additionally the fact that $\ad^* \fr \simeq S^2 \times so(2)^*$, we can trivialise also other tangents and cotangents of the Atiyah algebroid.

    Having all final trivialisations, we reduce the triple on the frame bundle. The canonical flip (relation) is given by
    \begin{gather*}
        \T\T S^2 \times so(2) \times so(2) \ni \left( \V, x, y \right) \sim_{\tilde{\kappa}_{A(\fr)}}  \left( \mathcal{U}, u, z \right) \in \T\T S^2 \times so(2) \times so(2) \\
        \iff \\
        \mathcal{U} = \kappa_{S^2}(\V) \quad \text{and} \quad z = y +  \langle \vec{w} \times \vec{v} \mid \vec{n} \rangle, \; \tau_{\T S^2}(\V) = (\vec{n}, \vec{v}), \; \T\tau_{S^2}(\V) = (\vec{n}, \vec{w}),
    \end{gather*}
    the Tulczyjew map by
    \begin{align*}
        \tilde{\varepsilon}_{A(\fr)} &: \T^*\T S^2 \times so(2) \times so(2)^* \longrightarrow \T\T^* S^2 \times so(2)^* \times so(2)^* \\
        &: (\omega, x, c) \longmapsto (\alpha_{S^2}^{-1}(\omega) \dot{-} c \langle \vec{v} \times \vec{n} \mid \pr_3(\cdot) \, \rangle, c, 0 ), \qquad \pi_{\T S^2}(\omega) = (\vec{n}, \vec{v}),
    \end{align*}
    and, finally, the symplectic-based map by
    \begin{align*}
        \tilde{\eta}_{A(\fr)} &: \T^*\T^* S^2 \times so(2)^* \times so(2) \longrightarrow \T\T^* S^2 \times so(2)^* \times so(2)^* \\
        &: (\Theta, a, z) \longmapsto (\beta_{S^2}^{-1}(\Theta) \dot{+} a \langle \vec{v} \times \vec{n} \mid \pr_3(\cdot) \, \rangle, a, 0 ), \qquad \xi_{\T S^2}(\Theta) = (\vec{n}, \vec{v}).
    \end{align*}
    The zeros in the last component are due to commutativity of $so(2)$, which yields $ \forall x \in so(2): \;\; \ad_x^* = 0$.

\subsection{Example: Reduction of the dynamics}
    We observe that the Lagrangian $L(p, (\vec{n}, \vec{v}), r) := \frac12 I_\perp \|\vec{n} \times \vec{v}\|^2 + \frac12 I_{ax}r^2$ is invariant under the tangent action \eqref{eq:tangent_action_SO2}. Therefore, we may introduce its reduced version (remembering that $\ad\fr \simeq S^2 \times so(2)$)
    \begin{align*}
        l(\vec{n}, \vec{v}, r) \equiv l((\vec{n}, \vec{v}), [(p, r)]_{\ad\fr}) := L(p, (\vec{n}, \vec{v}), r), \\
        \dif l (\vec{v}, \vec{n}, r) = ((\vec{n}, \vec{v}, c\vec{n}, I_\perp  \vec{v}), r, I_{ax}r).
    \end{align*}
    The dynamics reads
    \begin{align*}
        \widetilde{d} = \tilde{\varepsilon}_{A(\fr)} \circ \dif l (\T S^2 \times so(2)) = \left\{ \left((\vec{n}, I_\perp \vec{v}, \vec{v}, c\vec{n} - I_{ax}r\vec{v}\times\vec{n}), I_{ax}r, 0 \right) \right\}
    \end{align*}
    and, obviously, leads to the same Euler-Lagrange equations as the unreduced version \eqref{eq:E-L_frame}. 

\section{Acknowledgements}
    Parts of the results presented in this article were used in Kuba Krawczyk's and Paweł Korzeb's BSc theses at the University of Warsaw, Faculty of Physics. Both theses were supervised by Katarzyna Grabowska.

\end{document}